\documentclass[journal,twoside,web]{ieeecolor}
\usepackage{generic}
\usepackage{cite}
\usepackage{amsmath,amssymb,amsfonts}
\usepackage{graphicx}
\usepackage{textcomp}

\usepackage{array}
\usepackage{stfloats}
\usepackage{verbatim}
\usepackage{caption}
\usepackage{subcaption}
\newtheorem{theorem}{Theorem}

\def\BibTeX{{\rm B\kern-.05em{\sc i\kern-.025em b}\kern-.08em
    T\kern-.1667em\lower.7ex\hbox{E}\kern-.125emX}}
\begin{document}
\title{Distributed Differential Graphical Game for Control of Double-Integrator Multi-Agent Systems with Input Delay}
\author{Hossein B. Jond
%\thanks{This work was supported by the Czech Science Foundation (GACR) grant no. 23-07517S.}
\thanks{The author is with the Department of Cybernetics, Faculty of Electrical
Engineering, Czech Technical University in Prague, Czech Republic (e-mail:
hossein.jond@fel.cvut.cz).}}

\maketitle

\begin{abstract}
This paper studies cooperative control of noncooperative double-integrator multi-agent systems (MASs) with input delay on connected directed graphs in the context of a differential graphical game (DGG). In the distributed DGG, each agent seeks a distributed information control policy by optimizing an individual local performance index (PI) of distributed information from its graph neighbors. The local PI, which quadratically penalizes the agent's deviations from cooperative behavior (e.g., the consensus here), is constructed through the use of the graph Laplacian matrix. For DGGs for double-integrator MASs, the existing body of literature lacks the explicit characterization of Nash equilibrium actions and their associated state trajectories with distributed information. To address this issue, we first convert the $N$-player DGG with $m$ communication links into $m$ coupled optimal control problems (OCPs), which, in turn, convert to the two-point boundary-value problem (TPBVP). We derive the explicit solutions for the TPBV that constitute the explicit distributed information expressions for Nash equilibrium actions and the state trajectories associated with them for the DGG. An illustrative example verifies the explicit solutions of local information to achieve fully distributed consensus.  
\end{abstract}


\section{Introduction}
\IEEEPARstart{C}{ooperation} is fundamental in distributed biological multi-agent systems (MASs) for certain eco-evolutionary advantages~\cite{Nowak}. Cooperative human-engineered distributed MASs are a research trend that draws inspiration from biological systems like flocks of birds or bacterial colonies. In the past decades, cooperative control of MASs on graphs has received increasing attention due to its extensive applications such as consensus~\cite{LI20131986}, synchronization~\cite{6571228}, flocking~\cite{Olfati-Saber}, formation~\cite{Fu2014}, rendezvous~\cite{Rendezvous}, and so on. Classical control designs for such systems are centralized and require global knowledge of the system or a complete communication graph, which is excessive and conflicts with the communication infrastructure of MAS\cite{LIN2019266,Song}. Cooperative distributed control of MAS in terms of optimality is of prime interest. However, finding a locally distributed control that is optimal in some sense, e.g., when each agent is optimizing its own local performance index (PI), is particularly demanding~\cite{WanACC2021}. For noncooperative MASs on graphs, this can be achieved by a differential graphical game (DGG)~\cite{Vamvoudakis}, where the notion of optimality is Nash equilibrium.  

DGGs are a significant class of differential games with applications to the distributed control of networked systems on graphs. In recent years, there has been an increased interest in studying various cooperative control problems using differential (or dynamic) games (see, e.g.,consensus~\cite{7907492,ZHOU2022105376}, flocking~\cite{6623087}, synchronization~\cite{9565820,8232604}, formation~\cite{LIN2019266,9146203}, and pursuit-evasion~\cite{8754740,7126188}). The challenges of applying DGG-based control are twofold. First, differential games, in general, are difficult to solve~\cite{9599467,7909033,Vamvoudakis}. Second, even if a solution is found, it requires global information. However, there have been attempts to overcome these challenges, for instance, by estimation~\cite{Inga}, constructing local solutions to global Nash~\cite{Shi}, iterative control laws that converge to the Nash equilibrium~\cite{H-Li}, neural network-based integrated heuristic dynamic programming~\cite{Duan}, reinforcement learning for Nash~\cite{Ding}, value function approximation~\cite{Vamvoudakis}. Solving a DGG also depends on whether it is discrete-time~\cite{ZHOU2022105376} or Stackelberg type~\cite{Polycarpou}. This paper focuses on a continuous-time DGG of the Nash type.

Due to its analytical tractability, the framework of linear-quadratic (LQ) differential games is popular for analyzing MASs. The double-integrator MASs on graphs can be modeled as LQ DGGs with quadratic performance indices (PIs)~\cite{Gu,Lin2014,LIN2019266,9565820}. The emerging noncooperative DGG problem in this study is in fact a continuous-time LQ Nash differential game. The open-loop Nash equilibrium solution as well as its uniqueness and existence for this game are given in~\cite{Basar-book,Engwerda-book}. Nevertheless, the given Nash equilibrium solution requires global state information~\cite{Lin2014,LIN2019266}. The main challenge in noncooperative differential games is finding a distributed information Nash equilibrium solution that can be executed locally~\cite{Liu-Lopez,LIN2019266}. Recently, a distributed adaptive Nash equilibrium solution for DGGs with an infinite planning horizon was proposed~\cite{9565820}. In infinite-horizon differential games, there is no specified terminal time nor a terminal state constraint in the PIs associated with players, and the solution to the HJB equation for nonlinear systems or the Riccati equation for linear systems is time-invariant. In the case of a finite horizon, a terminal state constraint is usually included in the PIs, and the solution to either the HJB equation or the Riccati equation is time-dependent. Finding a distributed solution with a finite planning horizon is a great challenge. A distributed estimation of Nash equilibrium for DDGs with a finite planning horizon has been constituted in~\cite{LIN2019266}. On the other hand, explicit expressions have not been found except for differential games with simplified single-integrator dynamics in 1-dimensional coordinates (e.g., see~\cite{jond-IET,YildizO21,7058342}). The input delay considered in this paper even makes finding such a solution more difficult.  

Commonly, a time delay occurs in the control input when the control action depends on the relative state information transmitted over the wireless communication network. Input-delayed cooperative MASs have been extensively studied by the control community. However, only a few research works have been done in the game theoretical control of MASs with input delay (e.g.,~\cite{Ai}).

%The main contribution of the paper is deriving a distributed explicit expression of time, delay, and the initial state for a DGG with both a terminal PI of state and a running PI of state and control with input delay in comparison with terminal PI and only running PI of control without delay in~\cite{Lin2014,LIN2019266} or only running PI of state and control without a terminal PI and input delay in~\cite{jond-IET,YildizO21}. 

In this paper, cooperative control of double-integrator MASs on graphs is formulated as an LQ DGG. The Nash equilibrium of the DGG is a mutually beneficial control strategy that nobody will deviate from. Therefore, it can be exploited as a self-enforcing MAS control strategy. The main contribution of the paper is deriving a distributed explicit expression of time, delay, and the initial state for the LQ DGG with input delay. Our results differs from similar works in~\cite{Lin2014,LIN2019266,Engwerda-book,Gu,jond-IET,YildizO21,Ai} as follows:
\begin{enumerate}
  \item  The DGG in this paper has both a terminal PI of state and a running PI of state and control with input delay. The DGG in~\cite{Lin2014,LIN2019266} has a terminal PI and a running PI of only control without delay. The DGG in~\cite{jond-IET,YildizO21} has only running PI of state and control without delay and without terminal PI.  
  \item\cite{Engwerda-book,Gu,Ai} do not offer an explicit distributed solution for DGGs. While~\cite{Engwerda-book,Gu} presented a numerical non-distributed solution,~\cite{Ai} proposed a distributed Nash equilibrium-seeking algorithm with a distributed observer for MASs with input delay that converges asymptotically.    
\end{enumerate}

The rest of this paper is organized as follows. Section~\ref{sec:pre} provides preliminaries. The DGG for cooperative control of double-integrator multi-agent systems is introduced in Section~\ref{sec:DGG}. In Section~\ref{sec:main}, the introduced DGG is converted to a set of coupled optimal control problems (OCPs) where the explicit solution for the latter is derived. Section~\ref{sec:sim} illustrates the simulation results for the cooperative consensus control of double-integrator multi-agent systems with input delay. The conclusion is given in Section~\ref{sec:con}.

\section{Preliminaries}\label{sec:pre}

\subsection{Kronecker Product}
Let $X \in \mathit{\mathbb{R}}^{m\times n}, Y \in \mathit{\mathbb{R}}^{p\times q}$ be real-valued (or complex-valued) matrices. The Kronecker product is defined as
\begin{equation*}
X \otimes Y=
\begin{bmatrix}
x_{11}Y & \cdots & x_{1n}Y\\
\vdots & \ddots & \vdots\\
x_{m1}Y & \cdots & x_{mn}Y
\end{bmatrix}
\in \mathit{\mathbb{R}}^{mp \times nq}.         
\end{equation*}
The Kronecker product has the following properties \cite{laub2005matrix}
\begin{align}\label{eq:Kronecker}
&(X \otimes Y)^\top=X^\top \otimes Y^\top, \quad
(X \otimes Y)^{-1}=X^{-1} \otimes Y^{-1},\nonumber \\
&(X \otimes Y)(U \otimes V)=XU\otimes YV, \quad \mathrm{e}^{X \otimes Y}=\mathrm{e}^{Y} \otimes \mathrm{e}^{X},\nonumber\\
&\mathrm{det} (X \otimes Y)= (\mathrm{det} X)^m (\mathrm{det} Y)^n,  X \in \mathit{\mathbb{R}}^{n\times n}, Y\in \mathit{\mathbb{R}}^{m\times m}.
\end{align}

\subsection{Graph Theory}
A directed graph is a pair $\mathcal{G}(\mathcal{V},\mathcal{E})$ where $\mathcal{V}$ is a finite set of vertices/nodes and $\mathcal{E}\subseteq \{(i,j):i,j \in \mathcal{V}\}$ is a set of directed edges/arcs. Each edge $ (i,j) \in \mathcal{E}$ represents an information flow from node $i$ to node $j$ and is assigned a positive weight $\mu_{ij}>0$. The set of neighbors of vertex $i$ is defined by $\mathcal{N}_i=\{j \in \mathcal{V}:(i,j)~\text{or}~(j,i) \in \mathcal{E},j\neq i\}$. Graph $ \mathcal{G}$ is connected if for every pair of vertices $ (i,j) \in \mathcal{V}\times\mathcal{V}$, from $ i$ to $ j$ for all $ j  \in  \mathcal{V}, j \neq  i $, there exists a path of (undirected) edges from $\mathcal{E}$. 

Matrix $D \in \mathit{\mathbb{R}} ^{\mathcal{|V|}\times\mathcal{|E|}}$ is the incidence matrix of $\mathcal{G}$ where $D$'s $uv$th element is 1 if the node $u$ is the head of the edge $v$, $-1$ if the node $u$ is the tail, and $0$, otherwise. The distributed Laplacian matrix for each node $i$ that depends only on information about that node and its neighbors in the graph is defined as
\begin{equation*}
    L_i=DW_iD^\top \in\mathbb{R}^{\mathcal{|V|}\times \mathcal{|V|}}
\end{equation*}
where $W_i=\mathrm{diag}(\cdots,\mu_{ij},\cdots)~\forall j\in\mathcal{N}_i\in \mathit{\mathbb{R}}^{\mathcal{|E|}\times \mathcal{|E|}}$. 

The Kronecker product is used to extend the dimension of a matrix. The extended dimension Laplacian matrix $\hat{L}_i$ is defined as 
\begin{equation*}
\hat{L}_i=\hat{D}\hat{W}_i\hat{D}^\top\in\mathbb{R}^{m\mathcal{|V|}\times m\mathcal{|V|}}
\end{equation*}
where $\hat{D}=D\otimes I_{q\times q}$, $\hat{W}_i=W_i\otimes I_{q\times q}$, and $I_{q\times q}$ is the identity matrix of dimension $q\in\mathit{\mathbb{N}}$. 
The $m$-dimensional graph Laplacian $\hat{L}_i$ is symmetric $(\hat{L}_i=\hat{L}_i^\top)$, positive semidefinite $(\hat{L}_i\geq 0)$, and satisfies the (local) sum-of-squares property \cite{Olfati-Saber}
\begin{equation} \label{eq:sum-of-squares}
\sum_{\forall j\in\mathcal{N}_i} \mu_{ij} \| x_i-x_j \|^2 =x^\top \hat{L}_i x   
\end{equation}
where $ x=[x_1,\cdots,x_{\mathcal{|V|}}]^\top$ is a nonzero vector of all nodes states $x_i\in\mathbb{R}^{q\times 1}$ and $\lVert .\rVert$ is the Euclidean norm.

\section{Differential Graphical Game}\label{sec:DGG}

Consider a MAS $\mathcal{V}=\{0,1,\cdots,N\}$ with $N+1$ agents, where 0 corresponds to an externally controlled agent and the rest are agents with their control input to be designed. The inter-agent communications are restricted by a connected directed graph  $\mathcal{G}(\mathcal{V},\mathcal{E})$. Each agent $i\in\{1,\cdots,N\}$'s nonlinear dynamics is reduced to a double-integrator model as follows using the feedback linearization technique (e.g., see unmanned aerial vehicle (UAV) linearization in~\cite{Lin2014,6341809} and distributed generators (DGs) linearization in~\cite{9565820}) 
\begin{equation}\label{eq:double-int}
    \ddot{p}_i(t)=u_i(t-\tau)
\end{equation}
where $p_i\in\mathbb{R}^{q\times 1}$ and $u_i\in\mathbb{R}^{q\times 1}$ are the $q$-dimensional information and control input vectors, respectively, and $\tau>0$ is a constant time delay ($u_i(t)=0$ when $t<\tau$). The dynamics (\ref{eq:double-int}) is equivalent to the following linear time-invariant agent dynamics with input delay
\begin{equation}\label{eq:agent-dynamics}
\begin{bmatrix}
     \dot{p}_i(t)   \\ \ddot{p}_i(t)
\end{bmatrix}=(\begin{bmatrix}0 & 1\\ 0 & 0 \end{bmatrix}\otimes I_{q\times q})\begin{bmatrix}
    p_i(t) \\ \dot{p}_i(t)
\end{bmatrix}+(\begin{bmatrix}0 \\ 1 \end{bmatrix}\otimes I_{q\times q})u_i(t-\tau).
\end{equation}

%Let $n=|\mathcal{V}|$ and $m=|\mathcal{E}|$. 
Let $x(t)=[p_0^\top(t),\cdots,p_N^\top(t),\dot{p}_0^\top(t),\cdots,\dot{p}_N^\top(t)]^\top\in\mathbb{R}^{2(N+1)q\times 1}$ be the state vector for $\mathcal{V}$. The state dynamics of the system with a given initial state can then be described by
\begin{equation}\label{eq:MAS-dynamics}
    \dot{x}(t)=A_0 x(t)+\sum_{i=0}^N B_i u_i(t-\tau), \quad x(0)=x_0
\end{equation}
where $A_0=A\otimes I_{q\times q} \in\mathbb{R}^{2(N+1)q\times 2(N+1)q}$, $A=\begin{bmatrix}0 & 1\\ 0 & 0 \end{bmatrix}\otimes I_{(N+1)\times (N+1)} \in\mathbb{R}^{2(N+1)\times 2(N+1)}$, $B_i=\begin{bmatrix}
    0_{(N+1)\times 1} \\b_i
\end{bmatrix}\otimes I_{q\times q}\in\mathbb{R}^{2(N+1)q\times q}$, $b_i=[0,\cdots,1,\cdots,0]^\top\in\mathbb{R}^{(N+1)\times 1}$, and $0_{(N+1)\times 1}$ is the zero vector of dimension $N+1$.  

In the context of a DGG, each agent $i\in\{1,\cdots,N\}$ is so-called a player of the DGG, $\{1,\cdots,N\}$ corresponds to the set of players, and the system dynamics (\ref{eq:MAS-dynamics}) represents the state dynamics of the DGG. In the DGG approach to cooperative control of the noncooperative MAS (\ref{eq:MAS-dynamics}), each player seeks to optimize a local PI that solely reflects its interests by finding a suitable control policy. Define the finite horizon PI 
\begin{equation}
\label{eq:quadratic-cost}
J_i= C_i(x(T),T)+\int_{0}^T L_i(x(t),u_i(t-\tau),t)~\mathrm{dt}
\end{equation}
for each player $i\in\{1,\cdots,N\}$ where $C_i(x(T),T)$ is the terminal PI at horizon time $T$ and $L_i(x(t),u_i(t-\tau),t)$ is the running PI over the entire horizon. Note that $C_i$ and $L_i$ are associated with distributed information available only to agent $i$ locally through the graph $\mathcal{G}(\mathcal{V},\mathcal{E})$. Each player $i\in\{1,\cdots,N\}$ runs an optimization to find an admissible control policy $u_i(t-\tau)\in\mathcal{U}_i$ that optimizes its unified PI $J_i$ subject to the state dynamics (\ref{eq:MAS-dynamics}).  

Various cooperative control problems can be depicted by a distributed PI (\ref{eq:quadratic-cost}) that quadratically penalizes the local behavior of an agent. In this paper, we define $C_i(x(T),T)$ and $L_i(x(t),u_i(t-\tau),t)$ mathematically for consensus-seeking agents. Thereby, the vectors $p_i(t)$, $\dot{p}_i(t)$, and $u_i(t)$ represent the position, velocity, and acceleration input for agent $i$ in an $q$-dimensional space. Under the DGG framework, each player's aim is to achieve consensus while minimizing their PI $J_i$.   

The MAS in $\mathcal{V}$ governed by the state dynamics (\ref{eq:MAS-dynamics}) and the graph $\mathcal{G}(\mathcal{V},\mathcal{E})$ is said to achieve consensus if for any initial position $p_i(0)$ and initial velocity $\dot{p}_i(0)$, $\forall i\in\{1,\cdots,N\}$, 
\begin{equation} \label{eq:consensus}
   \sum_{j\in\mathcal{N}_i}\mu_{ij}\left(\|p_i(t)-p_j(t)\|^2+\|\dot{p}_i(t)-\dot{p}_j(t)\|^2\right)\rightarrow 0
\end{equation}
as $t \rightarrow T$ and $T>\tau$%is a finite horizon length
~\cite{1431045}. In the DGG context, a weighted sum of local consensus errors and local velocity errors at the horizon can be defined for each player $i\in\{1,\cdots,N\}$ as 
\begin{align} \label{eq:consensus-terminal}
   C_i&(x(T),T)=\nonumber\\&\sum_{j\in\mathcal{N}_i}\omega_{ij}\left(\|p_i(T)-p_j(T)\|^2+\|\dot{p}_i(T)-\dot{p}_j(T)\|^2\right)
\end{align}
where $\omega_{ij}>0$ is a scalar. In addition to the least (\ref{eq:consensus}) and (\ref{eq:consensus-terminal}), each player at the same time naturally seeks to expend the least control effort over the entire horizon. Therefore, a reasonable running PI for player $i\in\{1,\cdots,N\}$ is 
\begin{align}\label{eq:running-cost}
    L_i&(x(t),u_i(t-\tau),t)=r_{i}\lVert u_i(t-\tau) \rVert^2+ \nonumber\\&  \sum_{j\in\mathcal{N}_i}\mu_{ij}\left(\|p_i(t)-p_j(t)\|^2+\|\dot{p}_i(t)-\dot{p}_j(t)\|^2\right)  
\end{align}
where $r_{i}>0$ is a scalar penalizing the control effort term. 

The PIs (\ref{eq:consensus-terminal}) and (\ref{eq:running-cost}) can be depicted in compact form by using the sum-of-squares property (\ref{eq:sum-of-squares}) as follows
\begin{align} 
   C_i(x(T),T)&=x^\top (T) \big(I_{2\times 2}\otimes\hat{L}_{iT}\big) x(T),\label{eq:consensus-terminal-comp}\\
    L_i(x(t),u_i(t-\tau),t)&=x^\top (t)\big(I_{2\times 2}\otimes\hat{L}_i\big)x(t)+\nonumber \\
    &u_i^\top(t-\tau)R_{i}u_i(t-\tau)\label{eq:running-cost-comp}
\end{align}
where $\hat{L}_{iT}=\hat{D}\hat{W}_{iT}\hat{D}^\top$, $\hat{D}=D\otimes I_{q\times q}$, $\hat{W}_{iT}=W_{iT}\otimes I_{q\times q}$, $W_{iT}=\mathrm{diag}(\cdots,\omega_{ij},\cdots)~\forall j\in\mathcal{N}_i\in \mathit{\mathbb{R}}^{m\times m}$ (where $m=|\mathcal{E}|$), and $R_i=r_i\otimes I_{q\times q}$. Therefore, the minimization problem for each player $i\in\{1,\cdots,N\}$ consists of the unified PI (\ref{eq:quadratic-cost}) with components in compact form (\ref{eq:consensus-terminal-comp})-(\ref{eq:running-cost-comp}) subject to state dynamics (\ref{eq:MAS-dynamics}).

The requirement for consensus (\ref{eq:consensus}) can be easily generalized to most cooperative control problems, such as formation control. Accordingly, similar expressions to (\ref{eq:consensus-terminal}) and (\ref{eq:running-cost}) and thereby the compact forms (\ref{eq:consensus-terminal-comp}) and (\ref{eq:running-cost-comp}) can be acquired. For formation control, the compact form of PIs is given in~\cite{8945574}.

Each player's optimal strategy is the solution to its optimization problem. Under an open-loop information structure in the DGG (\ref{eq:MAS-dynamics}) and (\ref{eq:quadratic-cost}), all players simultaneously determine their strategies at the beginning of the game and use this open-loop strategy for the entire horizon. The players must then adjust their control inputs based on the delayed information. In a noncooperative DGG, the players of the game cannot make binding agreements. Therefore, the solution ought to be self-enforcing, such that, once agreed upon, no one has the incentive to deviate from it~\cite{van2012refinements}. A Nash equilibrium possesses the characteristic that any individual player's decision to deviate unilaterally does not result in a lower cost for that player. It is a strategy combination of all players in a noncooperative game where no one can achieve a lower cost by unilaterally deviating from it. The open-loop Nash equilibrium is defined as a set of admissible actions ($u_1^*,\cdots,u_N^* $) if for all admissible ($u_1,\cdots,u_N$) the following inequalities
\begin{align*}
   J_i (u_1^*,\cdots,u_{i-1}^*,u_i^*,u_{i+1}^*&,\cdots,u_N^* )\leq \\&
J_i (u_1^*,\cdots,u_{i-1}^*,u_i,u_{i+1}^*,\cdots,u_N^* ) 
\end{align*}
hold for $i\in\{1,\cdots,N\}$ where $u_i\in\Gamma_i$ and $\Gamma_i=\{u_i(t,x_0)|t\in[0,T]\}$ is the admissible strategy set for player $i$. We assume that $u_i$ ($\forall i\in\{1,\cdots,N\}$) consists of the set of measurable functions from $[0, T]$ into $\Gamma_i$ for which the differential equation (\ref{eq:MAS-dynamics}) has a unique solution and the PI (\ref{eq:quadratic-cost}) exists.

The noncooperative DGG problem in (\ref{eq:MAS-dynamics}) and (\ref{eq:quadratic-cost}) for $\tau=0$ admits a unique Nash equilibrium solution for every initial state $x_0$ iff there exists a solution set to a set of the coupled (asymmetric) Riccati differential equations (see \textit{Theorem 7.2} in~\cite{Engwerda-book}) or, equivalently, iff a specific matrix is invertible (see \textit{Theorem 7.1} in~\cite{Engwerda-book}). Moreover, Nash equilibrium requires global knowledge of the initial state vector $x(0)$. Therefore, the existence of a unique open-loop Nash equilibrium and whether or not the equilibrium and its associated state trajectories are distributed, as well as the input delay, are the main issues that need to be addressed for the DGG problem (\ref{eq:MAS-dynamics}) and (\ref{eq:quadratic-cost}).

\section{Main Results} \label{sec:main}
In this section, we present an explicit solution of distributed Nash actions and their associated state trajectories to the input delayed DGG problem in (\ref{eq:MAS-dynamics}) and (\ref{eq:quadratic-cost}). To the best of our knowledge, such solutions have not yet been reported for the LQ differential game. The presented solution results from applying the proposed systematic approach in Fig.~\ref{fig:diag}. The first step in this approach is converting the $N$-player DGG problem to $m$ OCPs. In terms of the MAS on the graph $\mathcal{G}(\mathcal{V},\mathcal{E})$, this is as converting the $(N+1)(=|\mathcal{V}|)$-node MAS to the $m(=|\mathcal{E}|)$-edge system. 

\begin{figure}[!t]
\centering
       \includegraphics[width=0.5\textwidth]{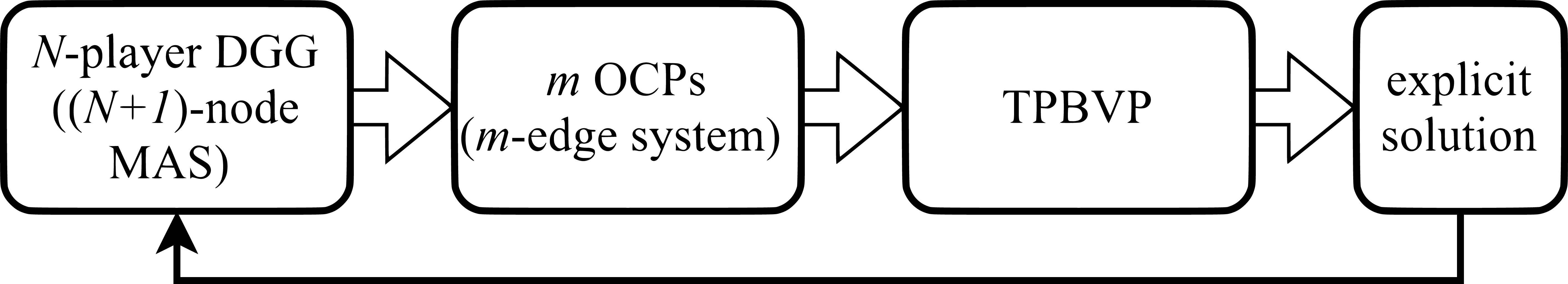}
\caption{Block diagram representation of the proposed systematic approach to solving the DGG.}
\label{fig:diag}
\end{figure} 

Toward this, we define the following new state vector and new control inputs 
\begin{align} 
    &z(t)=(I_{2\times 2}\otimes\hat{D}^\top)x(t)\in\mathbb{R}^{2Mq\times 1},\label{eq:new-state-vector} \\ 
    &\xi_i(t)=\begin{bmatrix}
0 & \cdots & I_{q\times q} & \cdots & 0
    \end{bmatrix}\hat{D}^\top
    \begin{bmatrix}
u_0(t) \\ \vdots \\ u_N(t) 
    \end{bmatrix}\in\mathbb{R}^{q\times 1}.\label{eq:new-control-vector}
\end{align}
The state dynamics (\ref{eq:MAS-dynamics}) in terms of the state vector $z(t)$ and control inputs $\xi_i(t)$ is given by
\begin{equation}\label{eq:new-game-dynamics}
\dot{z}(t)=\hat{A}_0 z(t)+\sum_{i=1}^m\hat{B}_i\xi_i(t-\tau), \quad z(0)=z_0
\end{equation}
where  $\hat{A}_0=\hat{A}\otimes I_{q\times q} \in\mathbb{R}^{2mq\times 2mq}$, $\hat{A}=\begin{bmatrix}0 & 1\\ 0 & 0 \end{bmatrix}\otimes I_{m\times m} \in\mathbb{R}^{2m\times 2m}$, $\hat{B}_i=\begin{bmatrix}
    0 \\ b_i
\end{bmatrix}\otimes I_{q\times q}\in\mathbb{R}^{2mq\times q}$, $\hat{b}_i=[0,\cdots,1,\cdots,0]^\top\in\mathbb{R}^{m\times 1}$, and $z_0=(I_{2\times 2}\otimes\hat{D}^\top)x_0$.%, and $\tau_i\neq 0$ is related to the pair node determined by the $i$th column of $\hat{\bm{D}}^\top$. 

The PIs (\ref{eq:consensus-terminal-comp}) and (\ref{eq:running-cost-comp}) are rewritten as follows
\begin{align} 
   &C_i(x(T),T)=x^\top (T) \big(I_{2\times 2}\otimes\hat{D}\hat{W}_{iT}\hat{D}^\top\big) x(T)\nonumber\\
   &=x^\top (T) \big(I_{2\times 2}\otimes\hat{D}\big)\big(I_{2\times 2}\otimes\hat{W}_{iT}\big)\big(I_{2\times 2}\otimes\hat{D}^\top\big) x(T)\nonumber\\
    &=\Big(\big(I_{2\times 2}\otimes\hat{D}^\top\big)x(T)\Big)^\top \big(I_{2\times 2}\otimes\hat{W}_{iT}\big)\big(I_{2\times 2}\otimes\hat{D}^\top\big) x(T)\nonumber\\
    &=z^\top (T)\big(I_{2\times 2}\otimes\hat{W}_{iT}\big) z(T)\triangleq \mathcal{C}_i(z(T),T)\label{eq:consensus-terminal-new},\\
    &\mathcal{L}_i(z(t),\xi_i(t-\tau),t)\triangleq \nonumber\\&z^\top(t) \big(I_{2\times 2}\otimes\hat{W}_i\big) z(t) 
    +\xi_i^\top(t-\tau)\hat{R}_{i}\xi_i(t-\tau)\label{eq:running-cost-new},
\end{align}
where $\hat{R}_{i}=\hat{r}_i\otimes I_{q\times q}$. Without loss of generality, we assume $r_i=r_j\forall(i,j)\in\mathcal{E}$, i.e., $\hat{r}_k=r_i=r_j$ for the $k$th edge $(i,j)\in\mathcal{E}$ in the $k$th column of $D$. As a result, the following PI emerges
\begin{equation}
\label{eq:new-quadratic-cost}
\mathcal{J}_i=\mathcal{C}_i(z(T),T)+\int_{0}^T\mathcal{L}_i(z(t),\xi_i(t-\tau),t)~\mathrm{dt}
\end{equation}
for $i\in\{1,\cdots,m\}$. The $N$-player DGG problem with input delay in (\ref{eq:MAS-dynamics}) and (\ref{eq:quadratic-cost}) is reduced to $m$ OCPs with input delay in (\ref{eq:new-game-dynamics}) and (\ref{eq:new-quadratic-cost}). In other words, the transition from the $(N+1)$-node MAS to the $m$-edge system, as shown in Fig.~\ref{fig:diag}, is complete. The latter is analytically more tractable to investigate for an explicit solution. 

The existence and uniqueness of the optimal control policies $\xi_i(t)$s and their associated state trajectory $z(t)$ for the $m$ OCPs (\ref{eq:new-game-dynamics}) and (\ref{eq:new-quadratic-cost}) can be investigated by applying the necessary conditions for optimality using Pontryagin’s principle. As a result, these problems are converted into the two-point boundary-value problem (TPBVP). In the context of the DGG (\ref{eq:MAS-dynamics}) and (\ref{eq:quadratic-cost}), the Laplacian matrices $\hat{L}_i$ and $\hat{L}_{iT}$ that occur in the performance index (\ref{eq:quadratic-cost}) are symmetric and positive semidefinite, and $R_i$ are positive definite. According to~\cite{ENGWERDA1998729}, then for the $N$-player differential game, the PIs $J_i$ ($i\in\{1,\cdots,N\}$) are strictly convex function of $u_i$ for all admissible control functions $u_j$, $j\neq i$, and for all $x_0$. This implies that the conditions obtained from Pontryagin’s principle are both necessary and sufficient for consensus control of the given MAS. Similarly, the matrices $\hat{W}_i$ and $\hat{W}_{iT}$ in the performance index (\ref{eq:new-quadratic-cost}) are also symmetric and positive semidefinite, and $\hat{R}_i$ are positive definite. Therefore, also for the $m$ OCPs, the conditions obtained from Pontryagin’s principle are both necessary and sufficient.

Before proceeding with applying the necessary conditions to the $m$ OCPs in order to obtain the TPBVP, one should address the input delay. A common approach in dealing with the input delayed linear dynamical system (\ref{eq:new-game-dynamics}) is to apply an integral transformation to obtain a delay-free linear dynamical system~\cite{delay2014}. Below, we convert the input delayed OCPs (\ref{eq:new-game-dynamics}) and (\ref{eq:new-quadratic-cost}) to the delay-free OCPs. 

The PI (\ref{eq:new-quadratic-cost}) can be decomposed by
\begin{align*}
\mathcal{J}_i=&\mathcal{C}_i(z(T),T)+\nonumber\\&\int_{0}^\tau\mathcal{L}_i(z(t),t)~\mathrm{dt}+\int_{0}^{T-\tau}\mathcal{L}_i(z(t+\tau),\xi_i(t),t)~\mathrm{dt}
\end{align*}
where the middle integral term on the right-hand side is a constant since the control does not take place for $t\in[0,\tau[$. Therefore, the minimization of $\mathcal{J}_i$ is equivalent to the minimization of 
\begin{align}\label{eq:transform-final}
\hat{\mathcal{J}}_i=\mathcal{C}_i(z(T),T)+\int_{0}^{T-\tau}\mathcal{L}_i(z(t+\tau),\xi_i(t),t)~\mathrm{dt}.
\end{align}

The solution of (\ref{eq:new-game-dynamics}) at $t+\tau$ is given by
\begin{align}
    z(t+\tau)&=\mathrm{e}^{(t+\tau)\hat{A}_0}z(0)\nonumber \\&+\int_0^{t+\tau}\mathrm{e}^{(t+\tau-s)\hat{A}_0}\sum_{i=1}^m\hat{B}_i\xi_i(s-\tau)~\mathrm{ds} \nonumber \\
  &=\mathrm{e}^{\tau\hat{A}_0}\Big(\mathrm{e}^{t\hat{A}_0}z(0)+\int_0^{t}\mathrm{e}^{(t-s)\hat{A}_0}\sum_{i=1}^m\hat{B}_i\xi_i(s-\tau)  ~\mathrm{ds}\nonumber \\&+\int_t^{t+\tau}\mathrm{e}^{(t-s)\hat{A}_0}\sum_{i=1}^m\hat{B}_i\xi_i(s-\tau)~\mathrm{ds}\Big).
\end{align}
The first and second terms inside the outer bracket pair on the right-hand side denote the solution of (\ref{eq:new-game-dynamics}) at $t$. Thereby, 
\begin{align}\label{eq:transform-dynamics}
    z(t+\tau)&=\mathrm{e}^{\tau\hat{A}_0}\left(z(t)+\int_t^{t+\tau}\mathrm{e}^{(t-s)\hat{A}_0}\sum_{i=1}^m\hat{B}_i\xi_i(s-\tau)~\mathrm{ds}\right) \nonumber \\&=\mathrm{e}^{\tau\hat{A}_0}\left(z(t)+\int_{t-\tau}^{t}\mathrm{e}^{(t-s-\tau)\hat{A}_0}\sum_{i=1}^m\hat{B}_i\xi_i(s)~\mathrm{ds}\right)\nonumber \\&\triangleq \mathrm{e}^{\tau\hat{A}_0} y(t).
\end{align}
By letting $t=0$ and $t=T-\tau$, we obtain  $\mathrm{e}^{-\tau\hat{A}_0}z(\tau)= y(0)$ and $z(T)= \mathrm{e}^{\tau\hat{A}_0}y(T-\tau)$, respectively. By substituting (\ref{eq:transform-dynamics}) into (\ref{eq:transform-final}), we have
\begin{align}\label{eq:transformed-PI}
\hat{\mathcal{J}}_i=\hat{\mathcal{C}}_i(y(T-\tau),T-\tau)+\int_{0}^{T-\tau}\hat{\mathcal{L}}_i(y(t),\xi_i(t),t)~\mathrm{dt}
\end{align}
where 
\begin{align*} 
  \mathcal{C}(z(T&),T)=z^\top (T)\big(I_{2\times 2}\otimes\hat{W}_{iT}\big) z(T)\nonumber \\
  &=\left(\mathrm{e}^{\tau\hat{A}_0}y(T-\tau) \right)^\top\big(I_{2\times 2}\otimes\hat{W}_{iT}\big) \left(\mathrm{e}^{\tau\hat{A}_0}y(T-\tau) \right)
  \nonumber\\
 &=y^\top(T-\tau) \left(\mathrm{e}^{\tau\hat{A}_0^\top}\big(I_{2\times 2}\otimes\hat{W}_{iT}\big) \mathrm{e}^{\tau\hat{A}_0}\right)y(T-\tau) \nonumber \\&
  \triangleq \hat{\mathcal{C}}_i(y(T-\tau),T-\tau)
\end{align*}
and  
\begin{align*}
   \mathcal{L}_i(z(&t+\tau),\xi_i(t),t)\nonumber\\&=z^\top(t+\tau) \big(I_{2\times 2}\otimes\hat{W}_i\big) z(t+\tau)+\xi_i^\top(t)\hat{R}_{i}\xi_i(t)\nonumber\\
    &=y^\top(t) \left(\mathrm{e}^{\tau\hat{A}_0^\top}\big(I_{2\times 2}\otimes\hat{W}_i\big) \mathrm{e}^{\tau\hat{A}_0}\right)y(t) +\xi_i^\top(t)\hat{R}_{i}\xi_i(t)\nonumber \\&\triangleq \hat{\mathcal{L}}_i(y(t),\xi_i(t),t).
\end{align*}

We notice that  $\hat{A}^2=0$ and $\hat{A}_0^2=(\hat{A}_0^\top)^2=0$. The matrix exponential terms can then be rewritten as
\begin{align*}
    \mathrm{e}^{\tau\hat{A}_0^\top}&\big(I_{2\times 2}\otimes\hat{W}_{iT}\big) \mathrm{e}^{\tau\hat{A}_0}=\nonumber\\&(I+\tau\hat{A}_0^\top)(I_{2\times 2}\otimes\hat{W}_{iT})(I+\tau\hat{A}_0)\triangleq \hat{Q}_{iT},\\
    \mathrm{e}^{\tau\hat{A}_0^\top}&\big(I_{2\times 2}\otimes\hat{W}_{i}\big) \mathrm{e}^{\tau\hat{A}_0}=\nonumber\\&(I+\tau\hat{A}_0^\top)(I_{2\times 2}\otimes\hat{W}_{i})(I+\tau\hat{A}_0)\triangleq \hat{Q}_{i}.
\end{align*}
Therefore, the PIs $\hat{\mathcal{C}}_i(y(T-\tau),T-\tau)$ and $\hat{\mathcal{L}}_i(y(t),\xi_i(t),t)$ in (\ref{eq:transformed-PI}) are simplified as 
\begin{align}
    \hat{\mathcal{C}}_i(y(T-\tau),T-\tau)=y^\top(T-\tau) \hat{Q}_{iT}y(T-\tau), \\
    \hat{\mathcal{L}}_i((t),\xi_i(t),t)=y^\top(t) \hat{Q}_{i}y(t) +\xi_i^\top(t)\hat{R}_{i}\xi_i(t).
\end{align}

The time derivative of (\ref{eq:transform-dynamics}) is given by
\begin{align}
    \mathrm{e}^{\tau\hat{A}_0}\dot{y}(t)=\dot{z}(t+\tau) 
    =\hat{A}_0 z(t+\tau)+\sum_{i=1}^m\hat{B}_i\xi_i(t)
\end{align}
or equivalently,
\begin{align}
    \dot{y}(t)=\mathrm{e}^{-\tau\hat{A}_0}\hat{A}_0 z(t+\tau)+\sum_{i=1}^m\hat{B}_{i0}\xi_i(t)
\end{align} 
where $\hat{B}_{i0}=\mathrm{e}^{-\tau\hat{A}_0}\hat{B}_i$. We also notice
\begin{align}\label{eq:simplification}
    \mathrm{e}^{\pm\tau\hat{A}_0}\hat{A}_0=  (\hat{A}_0\pm\tau\hat{A}_0^2)=\hat{A}_0\mathrm{e}^{\pm\tau\hat{A}_0}=\hat{A}_0.
\end{align}
Using the simplification above, we get
\begin{align}\label{eq:simple-linear}
    \dot{y}(t)&=\hat{A}_0 z(t+\tau)+\sum_{i=1}^m\hat{B}_{i0}\xi_i(t).
\end{align}
Substituting $z(t+\tau)=\mathrm{e}^{\tau\hat{A}_0} y(t)$ into (\ref{eq:simple-linear}) yields
\begin{align}
    \dot{y}(t)&=\hat{A}_0\mathrm{e}^{\tau\hat{A}_0} y(t)+\sum_{i=1}^m\hat{B}_{i0}\xi_i(t)
\end{align}
and re-substituting (\ref{eq:simplification}), finally the delay-free state dynamics is given by
\begin{align}\label{eq:delayfree-dynamics}
    \dot{y}(t)&=\hat{A}_0 y(t)+\sum_{i=1}^m\hat{B}_{i0}\xi_i(t),  y(0)=(I-\tau\hat{A}_0)z(\tau).
\end{align}
%with the given initial condition $y(0)=(I-\tau\hat{A}_0)z(\tau)$.

Now, we apply the conditions obtained from Pontryagin’s principle to the delay-free $m$ OCPs in (\ref{eq:delayfree-dynamics}) and (\ref{eq:transformed-PI}).  
According to Pontryagin’s principle, the Hamiltonian
\begin{align}
    H_i&=\hat{\mathcal{L}}_i(y(t),\xi_i(t),t) +\psi_i^\top(t) \dot{y}(t)%\nonumber\\    &=y^\top(t) Q_{i}y(t) +\xi_i^\top(t)\hat{R}_{i}\xi_i(t) +\psi_i^\top(t) \left(\hat{A}_0y(t)+\sum_{i=1}^m\hat{B}_{i0}\xi_i(t) \right)
\end{align}
is minimized %by player $i$ ($i=1,\cdots,m$) 
with respect to $\xi_i(t)$ (for $i=1,\cdots,m$). This gives the necessary conditions 
\begin{equation}\label{eq:neccesary-u}
   \xi_i(t)=-\hat{R}_i^{-1} B_{i0}^\top \psi_i(t) 
\end{equation}
where the vectors $\psi_i(t)$ satisfy
\begin{equation}\label{eq:co-estate}
    \dot{\psi}_i(t)=-\hat{Q}_i y(t)-\hat{A}_0^\top\psi_i(t), \quad \psi_i(T-\tau)=\hat{Q}_{iT}y(T-\tau).
\end{equation}
Substituting (\ref{eq:neccesary-u}) into (\ref{eq:delayfree-dynamics}) yields
\begin{equation}\label{eq:new-game-dynamics-}
\dot{y}(t)=\hat{A}_0 y(t)-\sum_{i=1}^m\hat{S}_i\psi_i(t), \quad y(0)=y_0
\end{equation}
where $\hat{S}_i=B_{i0}\hat{R}_i^{-1}B_{i0}^\top$.

The delay-free $m$ OCPs in (\ref{eq:delayfree-dynamics}) and (\ref{eq:transformed-PI}) have a solution iff the set of differential equations (\ref{eq:co-estate}) and (\ref{eq:new-game-dynamics-}) with the given boundary conditions above has a solution. Reference~\cite{Engwerda-book} (see \textit{Proof of Theorem 7.1}) shows that these differential equations with their boundary conditions convert to a TPBVP for a two-player differential game. This result can be generalized straightforwardly to the $m$ OCPs in (\ref{eq:delayfree-dynamics}) and (\ref{eq:transformed-PI}) as envisaged in Fig.~\ref{fig:diag}. The optimal control policy $\xi_i(t)$ as well as the associated state trajectory $y(t)$ can be calculated from the TPBVP. %If there is a solution to the OCPs, then for every $y_0$ this TPBVP has a solution. 

For the set of differential equations in (\ref{eq:co-estate}) and (\ref{eq:new-game-dynamics-}) for every initial state $y_0$, the associated state trajectory $y(t)$ with the optimal control policies is given by (see the \textit{proof of Theorem 7.4} in~\cite{Engwerda-book})
\begin{equation}\label{eq:TPBVP}
    y(t)=\hat{H}(T-\tau-t)\hat{H}^{-1}(T-\tau)y_0
\end{equation}
where
$\hat{M}=\begin{bmatrix}
-\hat{A}_0    & \hat{S}_1   &  \cdots & \hat{S}_m\\
\hat{Q}_1   & \hat{A}_0^\top  &    0   &    0\\
\vdots &      0 &    \ddots &    0\\
\hat{Q}_m & 0 & 0 & \hat{A}_0^\top
\end{bmatrix}$ and $\hat{H}(T-\tau-t)=\begin{bmatrix}
        I_{2mq\times 2mq} & 0 & \cdots & 0\end{bmatrix}\mathrm{e}^{(T-\tau-t)\hat{M}}\begin{bmatrix}
    I\\
    \hat{Q}_{1T}\\
    \vdots \\
    \hat{Q}_{mT} 
\end{bmatrix}$.

Using the Kronecker product property (\ref{eq:Kronecker}), we obtain
\begin{align}
    &\hat{Q}_{iT}\nonumber\\&=(I+\tau\hat{A}^\top\otimes I_{q\times q})(I_{2\times 2}\otimes W_{iT}\otimes I_{q\times q}\big)(I+\tau\hat{A}\otimes I_{q\times q})\nonumber\\
    &=\left((I+\tau\hat{A}^\top)(I_{2\times 2}\otimes W_{iT}\big)(I+\tau\hat{A}) \right)\otimes I_{q\times q}\nonumber\\&\triangleq Q_{iT}\otimes I_{q\times q},\label{eq:QiT}\\
    &\hat{Q}_{i}=\nonumber\\&\left((I+\tau\hat{A}^\top)(I_{2\times 2}\otimes W_{i}\big)(I+\tau\hat{A}) \right)\otimes I_{q\times q}\triangleq Q_{i}\otimes I_{q\times q},\label{eq:Qi}\\
  &\hat{S}_i=\Big(\begin{bmatrix}
      0 \\ \hat{b}_i
  \end{bmatrix}\otimes I_{q\times q}\Big)\Big(\frac{1}{ \hat{r}_i}\otimes I_{q\times q}\Big)\Big(\begin{bmatrix}
      0 & \hat{b}_i^\top
  \end{bmatrix}\otimes I_{q\times q}\Big)=\nonumber\\
  &
  \Big(\begin{bmatrix}
      0 \\ \hat{b}_i
  \end{bmatrix} \frac{1}{ \hat{r}_i} \begin{bmatrix}
      0 & \hat{b}_i^\top
  \end{bmatrix}\Big)\otimes I_{q\times q}=
  \begin{bmatrix}
      0 & 0 \\ 0 &  \frac{1}{\hat{r}_i}\hat{b}_i\hat{b}_i^\top
  \end{bmatrix} \otimes I_{q\times q}=\nonumber\\
  &\begin{bmatrix} 0 & 0\\  0 & 1 \end{bmatrix} \otimes
\frac{1}{\hat{r}_i}\hat{b}_i\hat{b}_i^\top \otimes I_{q\times q}\triangleq S_i\otimes I_{q\times q}\label{eq:S}.
\end{align}
%Let $\bm{\delta}_i=\frac{1}{\hat{r}_i}\hat{\bm{b}}_i\hat{\bm{b}}_i^\top=\mathrm{diag}(0,\cdots,\frac{1}{\hat{r}_i},\cdots,0)\in\mathbb{R}^{m\times m}$ and $\bm{S}_i=\begin{bmatrix} 0 & 0\\  0 & 1 \end{bmatrix} \otimes \bm{\delta}_i$ so  $\hat{\bm{S}}_i=\bm{S}_i\otimes\bm{I}_{q\times q}$.
Using (\ref{eq:QiT})-(\ref{eq:S}), the matrices $\hat{M}$ and $\hat{H}(T-\tau)$ can be expanded as $\hat{M}=M\otimes I_{q\times q}$ and $\hat{H}(T-\tau-t)=H(T-\tau-t)\otimes I_{q\times q}$, respectively, where
$M=\begin{bmatrix}
-\hat{A}    & S_1   &  \cdots & S_m\\
Q_1   & \hat{A}^\top  &    0   &    0\\
\vdots &      0 &    \ddots &    0\\
Q_m & 0 & 0 & \hat{A}^\top
\end{bmatrix}$ and $H(T-\tau-t)=\begin{bmatrix}
        I_{2m\times 2m} & 0 & \cdots & 0\end{bmatrix}\mathrm{e}^{(T-\tau-t)M}\begin{bmatrix}
    I\\
   Q_{1T}\\
    \vdots \\
   Q_{mT} 
\end{bmatrix}$.

Therefore, (\ref{eq:TPBVP}) is rewritten as
\begin{align}\label{eq:TPBVP-ext}
    y(t)&=\big(H(T-\tau-t)\otimes I_{q\times q}\big)\big(H^{-1}(T-\tau)\otimes I_{q\times q}\big)y_0 \nonumber\\
    &=\big(H(T-\tau-t)H^{-1}(T-\tau)\otimes I_{q\times q}\big)y_0.
\end{align}
%For the differential game defined in (\ref{eq:new-game-dynamics}) and (\ref{eq:new-quadratic-cost}) for every initial state $y_0$, the associated state trajectory $y(t)$ with the unique Nash equilibrium actions.

From (\ref{eq:TPBVP-ext}), it is obvious that the existence of a unique solution to the TPBVP depends on whether or not $H(T-\tau)$ is invertible. If invertible, there exists a unique open-loop Nash equilibrium for the DGG problem. Besides, if the individual trajectories $y_i(t)$ in (\ref{eq:TPBVP-ext}) do not require knowledge of $y_j(0),j\neq i$ from $y_0$, one can conclude that the equilibrium and its associated state trajectories for the DGG problem are distributed. In the following, we prove the invertibility of $H(T-\tau)$ i.e., the existence of unique optimal control policies and present individual state trajectories $y_i(t)$ in the form of explicit expressions of time, delay, and the initial state $y_i(0)$. Before, the following definitions were given.

Let $\lambda_i$ denote the $i$th eigenvalue of $M$. Define $f(\lambda_i)\triangleq f_1(\lambda_i)+f_2(\lambda_i)$ where $f_1(\lambda_i)\triangleq(\frac{\omega_i}{\hat{r}_i}+\frac{\mu_i}{\hat{r}_i}\tau)\frac{1}{\lambda_i^2}-\lambda_i$ and $f_2(\lambda_i)\triangleq(\frac{\omega_i}{\hat{r}_i}\tau+\frac{\mu_i}{\hat{r}_i}(\tau^2+1))\frac{1}{\lambda_i}$. Also, $g(\lambda_i)\triangleq g_1(\lambda_i)+g_2(\lambda_i)$
where $g_1(\lambda_i)\triangleq\frac{\omega_i}{\hat{r}_i}\tau\frac{1}{\lambda_i^2}+1$ and $g_2(\lambda_i)\triangleq\frac{\omega_i}{\hat{r}_i}(\tau^2+1)\frac{1}{\lambda_i}$. Notice that $f_1(-\lambda_i)=f_1(\lambda_i)$, $f_2(-\lambda_i)=-f_2(\lambda_i)$, $g_1(-\lambda_i)=g_1(\lambda_i)$, and $g_2(-\lambda_i)=-g_2(\lambda_i,\tau)$. Define 
\begin{equation}\label{eq:varpi}
    \sigma_i(\lambda_i)=[0,\cdots,\varpi_i(\lambda_i),\cdots,0]^\top,  \varpi_i^2(\lambda_i)=\frac{1}{2}\frac{\lambda_i^3}{\frac{\mu_i}{r_i} -\lambda_i^4}.
  \end{equation}

\begin{theorem}\label{theorem} $H(T-\tau)$ is nonsingular and 
\begin{align}\label{eq:trajectory}
y_i(t)&=\left(\frac{1}{\Delta_i(T-\tau)}\begin{bmatrix}
    q_i(t) & \hat{q}_i(t) \\
   \Tilde{q}_i(t) & \bar{q}_i(t)     
\end{bmatrix}\otimes I_{q\times q}\right)y_i(0)
\end{align}
where 
\small
\begin{align*}
   &\Delta_i(T-\tau)= \beta_i(T-\tau)\gamma_i(T-\tau)-\alpha_i(T-\tau)\eta_i(T-\tau),\\
   &q_i(t)=\beta_i(T-\tau-t)\gamma_i(T-\tau)-\alpha_i(T-\tau-t)\eta_i(T-\tau),\\
   &\hat{q}_i(t)=\alpha_i(T-\tau-t)\beta_i(T-\tau)-\beta_i(T-\tau-t)\alpha_i(T-\tau),\\
   &\Tilde{q}_i(t)=\eta_i(T-\tau-t)\gamma_i(T-\tau)-\gamma_i(T-\tau-t)\eta_i(T-\tau),\\
   &\bar{q}_i(t)=\gamma_i(T-\tau-t)\beta_i(T-\tau)-\eta_i(T-\tau-t)\alpha_i(T-\tau),\\
&\alpha_i(\phi)=4\Re(\varpi_i^2(\lambda_i)\Big[f_1(\lambda_i)\sinh(\phi\lambda_i)+f_2(\lambda_i)\cosh(\phi\lambda_i)\Big]),  \\
&\beta_i(\phi)=4\Re(\varpi_i^2(\lambda_i)\Big[g_1(\lambda_i)\sinh(\phi\lambda_i)+g_2(\lambda_i)\cosh(\phi\lambda_i)\Big]), \\
&\gamma_i(\phi)=-4\Re(\lambda_i\varpi_i^2(\lambda_i)\Big[f_1(\lambda_i)\sinh(\phi\lambda_i)+f_2(\lambda_i)\cosh(\phi\lambda_i)\Big]), \\
&\eta_i(\phi)=-4\Re(\lambda_i\varpi_i^2(\lambda_i)\Big[g_1(\lambda_i)\sinh(\phi\lambda_i)+g_2(\lambda_i)\cosh(\phi\lambda_i)\Big]). 
\end{align*}
\normalsize
\end{theorem}

\begin{proof}
To extract the explicit expressions for $y_i(t)$ from (\ref{eq:TPBVP-ext}), one has to find the inverse of $H(\phi)$. The invertibility of $H(\phi)$ depends on $M$ and $\phi$. We begin with the decomposition of the square matrix $M$ into the product of matrices. Matrix $M$ is defective (i.e., non-diagonalizable, see Appendix~\ref{app:defective}). Its Jordan decomposition is
\begin{equation}\label{eq:jordan}
M= \Phi J \Psi                         
\end{equation}
where $\Phi$, $J$, and $\Psi$ are square matrices given in Appendix~\ref{app:decompose}. 

The matrix exponential term $\mathrm{e}^{\phi M}$ in $H(\phi)$ is expanded as 
\begin{equation}\label{eq:matrix-exponential}
    \mathrm{e}^{\phi M}=\Phi \mathrm{e}^{\phi J} \Psi.
\end{equation}
Using this expansion, as shown in Appendix~\ref{app:decomposeH}, we have  
\begin{align}\label{eq:expanH}
   H(\phi)=\sum_{i=1}^{m}\Big(K(\lambda_i)+K(\overline{\lambda}_i)+K(-\lambda_i)+K(-\overline{\lambda}_i)\Big)
\end{align}
where
\begin{align}\label{eq:matrixK}
   K(\lambda_i)&=\mathrm{e}^{\phi\lambda_i}\begin{bmatrix}
  f(\lambda_i)
  &  
 g(\lambda_i) \\
   -\lambda_i f(\lambda_i)
  &  
  -\lambda_i g(\lambda_i)
 \end{bmatrix} \otimes \sigma_i(\lambda_i) \sigma_i^\top(\lambda_i)
\end{align}
and $f(\lambda_i)$, $g(\lambda_i)$, and $\sigma_i(\lambda_i)$ are defined beforehand. 

As $\overline{\lambda}_i$ are $\lambda_i$ the complex conjugate of each other, we have
\begin{align}
  \Re (K(\overline{\lambda}_i))=\Re(K(\lambda_i)),\quad
  \Im (K(\overline{\lambda}_i))=-\Im(K(\lambda_i))
\end{align}
where $\Re(.) $ and $\Im(.) $ denote the real and imaginary parts, respectively. Thereby, (\ref{eq:expanH}) is simplified further as
\begin{equation}
H(\phi)=2 \sum_{i=1}^{m} \big(\Re(K(\lambda_i))+\Re(K(-\lambda_i))\big).
\end{equation}
Therefore, $H(\phi)$ is a real-valued matrix. 

From (\ref{eq:matrixK}) we can see that $K(\lambda_i)$  has the structure (\ref{eq:matrixKK}), shown at the bottom of the page. 
\begin{figure*}[!b]%[t!]
\hrule
\begin{align}\label{eq:matrixKK}
   K(\lambda_i)&=\begin{bmatrix}
  \mathrm{diag}(0,\cdots,\mathrm{e}^{\phi\lambda_i}\varpi_i^2(\lambda_i)f(\lambda_i),\cdots,0) &  \mathrm{diag}(0,\cdots,\mathrm{e}^{\phi\lambda_i}\varpi_i^2(\lambda_i)g(\lambda_i),\cdots,0) \\
  \mathrm{diag}(0,\cdots,-\lambda_i \mathrm{e}^{\phi\lambda_i}\varpi_i^2(\lambda_i)f(\lambda_i),\cdots,0) &  \mathrm{diag}(0,\cdots,-\lambda_i \mathrm{e}^{\phi\lambda_i}\varpi_i^2(\lambda_i)g(\lambda_i),\cdots,0)  
\end{bmatrix}
\end{align}
\end{figure*}
On this basis, $H(\phi)$ has the following form 
\begin{align}\label{eq:matrixH}
   H(\phi)&=\begin{bmatrix}
  \mathrm{diag}(\cdots,\alpha_i(\phi),\cdots) &  \mathrm{diag}(\cdots,\beta_i(\phi),\cdots) \\
  \mathrm{diag}(\cdots,\gamma_i(\phi),\cdots) &  \mathrm{diag}(\cdots,\eta_i(\phi),\cdots)  
\end{bmatrix}
\end{align}
where $\alpha_i(\phi)$, $\beta_i(\phi)$, $\gamma_i(\phi)$, $\eta_i(\phi)$ are given by (\ref{eq:alpha-formula})-(\ref{eq:eta-formula}), shown at the bottom of the next page.
\begin{figure*}[!b]%[t!]
\begin{align}
\alpha_i(\phi)=&\mathrm{e}^{\phi\lambda_i}\varpi_i^2(\lambda_i)f(\lambda_i)+\mathrm{e}^{-\phi\lambda_i}\varpi_i^2(-\lambda_i)f(-\lambda_i)+\mathrm{e}^{\phi\overline{\lambda}_i}\varpi_i^2(\overline{\lambda}_i)f(\overline{\lambda}_i)+\mathrm{e}^{-\phi\overline{\lambda}_i}\varpi_i^2(-\overline{\lambda}_i)f(-\overline{\lambda}_i) \label{eq:alpha-formula},\\
    \beta_i(\phi)=&\mathrm{e}^{\phi\lambda_i}\varpi_i^2(\lambda_i)g(\lambda_i)+\mathrm{e}^{-\phi\lambda_i}\varpi_i^2(-\lambda_i)g(-\lambda_i)+\mathrm{e}^{\phi\overline{\lambda}_i}\varpi_i^2(\overline{\lambda}_i)g(\overline{\lambda}_i,\tau)+\mathrm{e}^{-\phi\overline{\lambda}_i}\varpi_i^2(-\overline{\lambda}_i)g(-\overline{\lambda}_i) ,\\
    \gamma_i(\phi)=&-\lambda_i\mathrm{e}^{\phi\lambda_i}\varpi_i^2(\lambda_i) f(\lambda_i)-\lambda_i\mathrm{e}^{-\phi\lambda_i} \varpi_i^2(-\lambda_i)f(-\lambda_i)-\overline{\lambda}_i\mathrm{e}^{\phi\overline{\lambda}_i}\varpi_i^2(\overline{\lambda}_i) f(\overline{\lambda}_i)-\overline{\lambda}_i\mathrm{e}^{-\phi\overline{\lambda}_i} \varpi_i^2(-\overline{\lambda}_i)f(-\overline{\lambda}_i), \\
    \eta_i(\phi)=&-\lambda_i \mathrm{e}^{\phi\lambda_i}\varpi_i^2(\lambda_i)g_i(\lambda_i)- \lambda_i\mathrm{e}^{-\phi\lambda_i}\varpi_i^2(-\lambda_i)g_i(-\lambda_i)-\overline{\lambda}_i \mathrm{e}^{\phi\overline{\lambda}_i}\varpi_i^2(\overline{\lambda}_i)g_i(\overline{\lambda}_i)- \overline{\lambda}_i\mathrm{e}^{-\phi\overline{\lambda}_i}\varpi_i^2(-\overline{\lambda}_i)g_i(-\overline{\lambda}_i)  \label{eq:eta-formula}
\end{align}
\end{figure*}

To find its inverse, first, we show that $H(\phi)$ is nonsingular. The eigenvalues of $H(\phi)$ are the roots of its characteristic polynomial 
\begin{align*}
 &\rho(\delta)= \det\big(H(\phi)-\delta I\big) =\det \Big(\mathrm{diag}(\cdots,\\
 &\big(\alpha_i(\lambda_i)-\delta\big)\big(\eta_i(\lambda_i)-\delta\big)-\beta_i(\lambda_i)\gamma_i(\lambda_i),\cdots)\Big) =\prod_{i=1}^{m}\Big(\delta^2 -\\
 &\big(\alpha_i(\lambda_i)+\eta_i(\lambda_i)\big)\delta + \alpha_i(\lambda_i)\eta_i(\lambda_i)- \beta_i(\lambda_i)\gamma_i(\lambda_i)\Big).
\end{align*}
None of the roots of the quadratic expression above are zero, meaning that $H(\phi)$ has no zero eigenvalues and is nonsingular. %Thus, the state trajectory $y(t)$ exists and therefore the optimal control actions $\xi$s and vice versa. 
As its blocks (which are diagonal matrices) are commuting matrices, the inverse of $H(\phi)$ is obtained by matrix analyses as (\ref{eq:inverseH}), shown at the bottom of the next page.
\begin{figure*}[!b]%[t!]
\begin{align}  \label{eq:inverseH} 
   H^{-1}(\phi)=\begin{bmatrix}
  \mathrm{diag}(\cdots,-\frac{\eta_i(\phi)}{\beta_i(\phi)\gamma_i(\phi)-\alpha_i(\phi)\eta_i(\phi)},\cdots) &  \mathrm{diag}(\cdots,\frac{\beta_i(\phi)}{\beta_i(\phi)\gamma_i(\phi)-\alpha_i(\phi)\eta_i(\phi)},\cdots) \\
  \mathrm{diag}(\cdots,\frac{\gamma_i(\phi)}{\beta_i(\phi)\gamma_i(\phi)-\alpha_i(\phi)\eta_i(\phi)},\cdots) &  \mathrm{diag}(\cdots,-\frac{\alpha_i(\phi)}{\beta_i(\phi)\gamma_i(\phi)-\alpha_i(\phi)\eta_i(\phi)},\cdots)  
\end{bmatrix}.
\end{align}
\end{figure*}

Knowing that $\Re(\mathrm{e}^{\phi\overline{\lambda}_i})=\Re(\mathrm{e}^{\phi\lambda_i})$, $\Re(\varpi_i^2(\overline{\lambda}_i))=\Re(\varpi_i^2(\lambda_i))$, and $\Re(f(\overline{\lambda}_i))=\Re(f(\lambda_i))$, finally, we can express $\alpha_i(\phi)$, $\beta_i(\phi)$, $\gamma_i(\phi)$, $\eta_i(\phi)$ in terms of the hyperbolic sine and cosine in (\ref{eq:trajectory}). For $\alpha_i(\phi)$, it is given in (\ref{eq:alpha-simplification}), shown at the bottom of the next page, and $\beta_i(\phi)$, $\gamma_i(\phi)$, and $\eta_i(\phi)$ are obtained similarly.
\begin{figure*}[!b]%[t!]   
\hrule
\begin{align}\label{eq:alpha-simplification}
\alpha_i(\phi)=&2\Re(\mathrm{e}^{\phi\lambda_i})\Re(\varpi_i^2(\lambda_i))\Re(f(\lambda_i))+2\Re(\mathrm{e}^{-\phi\lambda_i})\Re(\varpi_i^2(-\lambda_i))\Re(f(-\lambda_i))
\nonumber\\
    =&2\Re(\mathrm{e}^{\phi\lambda_i})\Re(\varpi_i^2(\lambda_i))\Big(\Re(f_1(\lambda_i))+\Re(f_2(\lambda_i))\Big)-2\Re(\mathrm{e}^{-\phi\lambda_i})\Re(\varpi_i^2(\lambda_i))\Big(\Re(f_1(\lambda_i))-\Re(f_2(\lambda_i))\Big)\nonumber \\
    =&4\Re(\varpi_i^2(\lambda_i))\Big[\Re(f_1(\lambda_i))\frac{\Re(\mathrm{e}^{\phi\lambda_i})-\Re(\mathrm{e}^{-\phi\lambda_i})}{2}+\Re(f_2(\lambda_i))\frac{\Re(\mathrm{e}^{\phi\lambda_i})+\Re(\mathrm{e}^{-\phi\lambda_i})}{2}\Big]\nonumber \\
    =&4\Re(\varpi_i^2(\lambda_i))\Big[\Re(f_1(\lambda_i,\tau))\Re(\sinh(\phi\lambda_i))+\Re(f_2(\lambda_i))\Re(\cosh(\phi\lambda_i))\Big]
\end{align}
\end{figure*}
Finally, multiplying $H(T-\tau-t)$ from (\ref{eq:matrixH}) to $H^{-1}(T-\tau)$ from (\ref{eq:inverseH}) yields (\ref{eq:trajectory}) and the proof is concluded.  
\end{proof}

For $t<\tau$, $\xi_i(t)=0$ for all $i\in\{1,\cdots,m\}$ and $z(t)=(I+t \hat{A}_0)y(0)$. For $t\geq \tau$, from (\ref{eq:transform-dynamics}), $z(t+\tau)=(I+\tau \hat{A}_0)y(t)$. Equivalently,  
\begin{align}
&z_i(t)=\left(\begin{bmatrix} 1 & t \\ 0 & 1\end{bmatrix}\otimes I_{q\times q}\right)y_i(0), \quad t<\tau,\\
&z_i(t+\tau)=\left(\begin{bmatrix} 1 & \tau \\ 0 & 1\end{bmatrix}\otimes I_{q\times q}\right)y_i(t), \quad t\geq \tau.
\end{align}

The control inputs $\xi_i(t)$ for $t\geq \tau$ are obtained from the double-integrator relations,  
\begin{align}\label{eq:equilibrium}
&\xi_i(t+\tau)=\begin{bmatrix}0&I_{q\times q}
\end{bmatrix}\dot{y}_i(t), \quad t\geq \tau. 
\end{align}

As shown in Fig.~\ref{fig:diag}, the proposed systematic approach to solving the DGG problem accomplishes this by using the optimal control policies $\xi_i(t)$ and their corresponding state trajectories $z_k(t)$ for $\forall i\in\{1,\cdots,m\}$ in the $m$-edge system to find the distributed explicit open-loop Nash equilibrium actions $u_j(t)$ and their corresponding distributed state trajectories $x_j(t)$ for $\forall j\in\{1,\cdots,N\}$ in the $N$-player DGG.

\section{Illustrative Example}\label{sec:sim}

In this section, we present an illustrative example to verify the distributed solution given in \textit{Theorem~\ref{theorem}}. Assume a consensus-seeking MAS $\mathcal{V}=\{0,1,2,3\}$ and $\mathcal{E}=\{(0,1),(1,2),(1,3)\}$. The communication graph is not complete, and agents 2 and 3 cannot acquire global knowledge of the initial state vector $x_0$.

The lead agent is determined by $p_0(t)=[\cos(t),t]^\top$, $\dot{p}_0(t)=[-\sin(t),1]^\top$, and $\dot{u}_0(t)=[-\cos(t),0]^\top$. The initial states of the other agents are $p_1(0)=[-1,1]^\top$, $\dot{p}_1(0)=[0,2]^\top$, $p_2(0)=[4,4]^\top$, $\dot{p}_2(0)=[0,0]^\top$, $p_3(0)=[6,9]^\top$, and $\dot{p}_3(0)=[2,0]^\top$. In the PIs, $W_{1T}=W_{2T}=W_{3T}=I_{3\times 3}$, $W_1=\mathrm{diag}(1,0.7,0.5)$, $W_2=\mathrm{diag}(0,0.7,0)$, $W_3=\mathrm{diag}(0,0,0.5)$, $r_1=r_2=r_3=1$, and $T=8$. 

Following the proposed approach in Fig.~\ref{fig:diag}, the $3$-player DGG converts to the TPBVP for which the explicit solution is given in (\ref{eq:trajectory}). The distributed Nash strategies and their associated distributed state trajectories are $u_1(t)=u_0(t)+\xi_1(t)$, $u_2(t)=\xi_1(t)+u_1(t)$, $u_3(t)=\xi_2(t)+u_1(t)$, $x_1(t)=x_0(t)+z_1(t)$, $x_2(t)=x_1(t)+z_2(t)$, $x_3(t)=x_2(t)+z_3(t)$ where $x_i(t)=[p_i^\top(t),\dot{p}_i^\top(t)]^\top$. For comparison, we employ the non-distributed solution (\ref{eq:TPBVP-ext}). Figs.~\ref{fig:traj}-\ref{fig:histdelay} show the agents' trajectories and time histories generated by both the distributed and non-distributed solutions for delay-free DGG and delayed DGG. It is seen that the agents $\{1,2,3\}$' positions, velocities, and control inputs reach a consensus with the lead agent. 

From Figs.~\ref{fig:traj}-\ref{fig:histdelay}, we observe that the behavior of agent $1$ is identical under both the non-distributed and distributed solutions. Notice from the communication graph $\mathcal{E}$ that agent $1$ has access to the global knowledge of all other agents, i.e., $\{0,2,3\}$, while they do not. As seen, agents $2$ and $3$ have different behaviors under both solutions. When they implement the distributed solution, these two agents converge to the lead agent's trajectory on a shorter path than when they implement the non-distributed solution.

\begin{figure}[ht]
   \centering
     \begin{subfigure}[b]{0.4\textwidth}
         \centering
         \includegraphics[width=\textwidth]{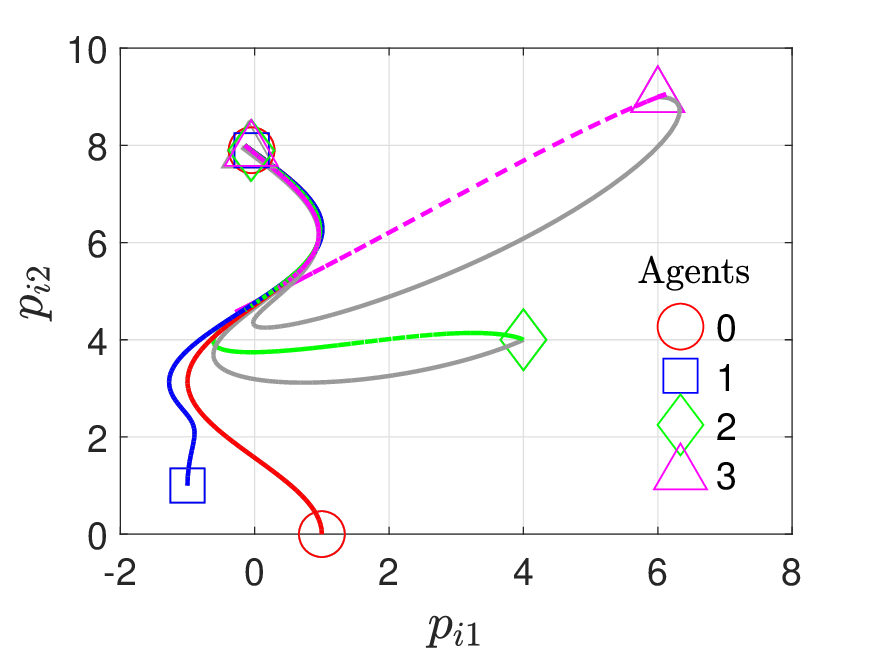}
         \caption{delay-free trajectories}
         %\label{fig:sim1}
     \end{subfigure}
     \begin{subfigure}[b]{0.4\textwidth}
         \centering
         \includegraphics[width=\textwidth]{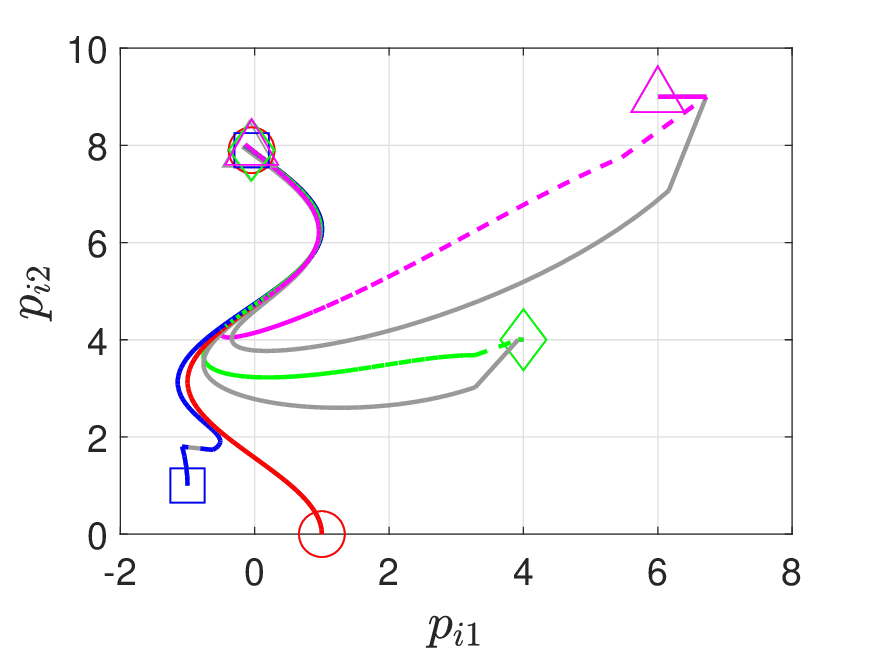}
         \caption{delayed ($\tau=0.5$) trajectories}
         %\label{fig:sim2}
     \end{subfigure}
\caption{Agents' trajectories under the non-distributed (solid gray lines) and distributed solutions (dashed color lines).} 
\label{fig:traj}
\end{figure}

\begin{figure*}[ht]
   \centering
   \includegraphics[width=0.25\textwidth]{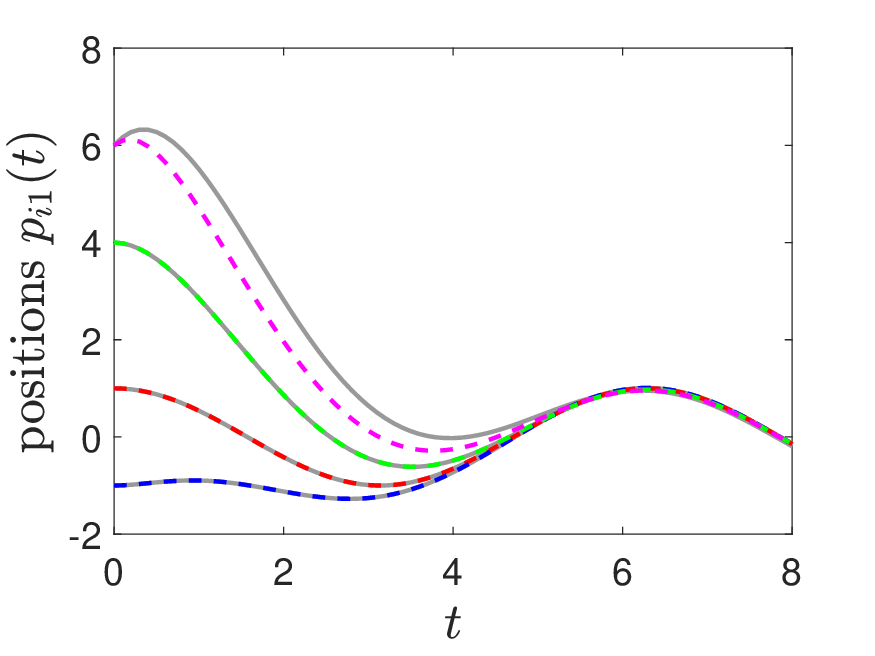}
   \includegraphics[width=0.25\textwidth]{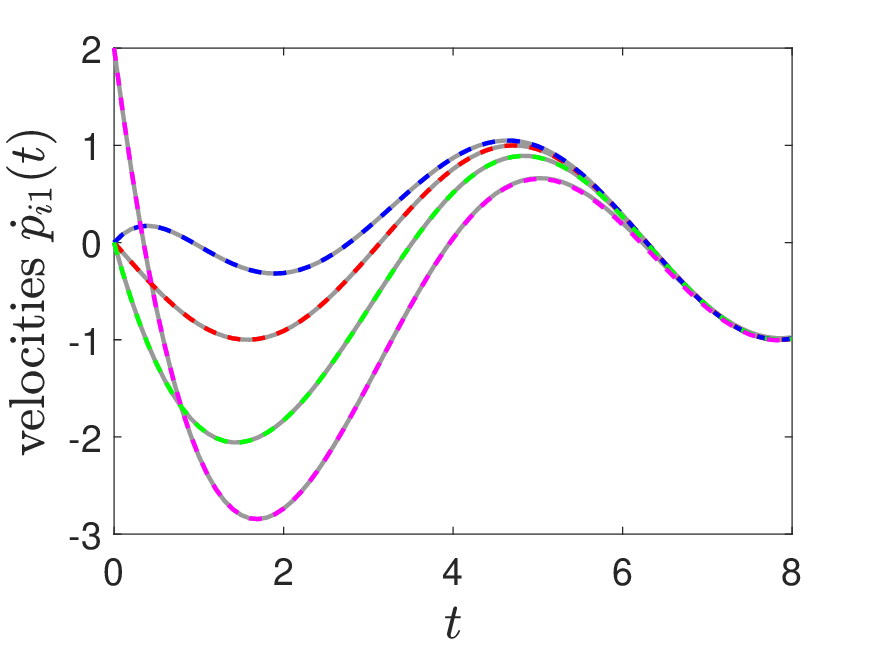}
   \includegraphics[width=0.25\textwidth]{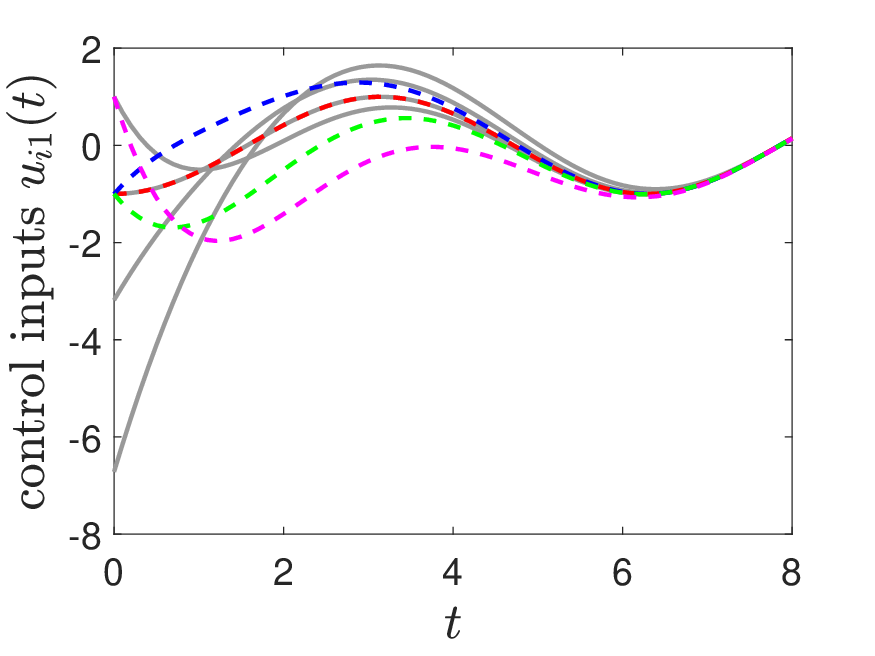}
   \includegraphics[width=0.25\textwidth]{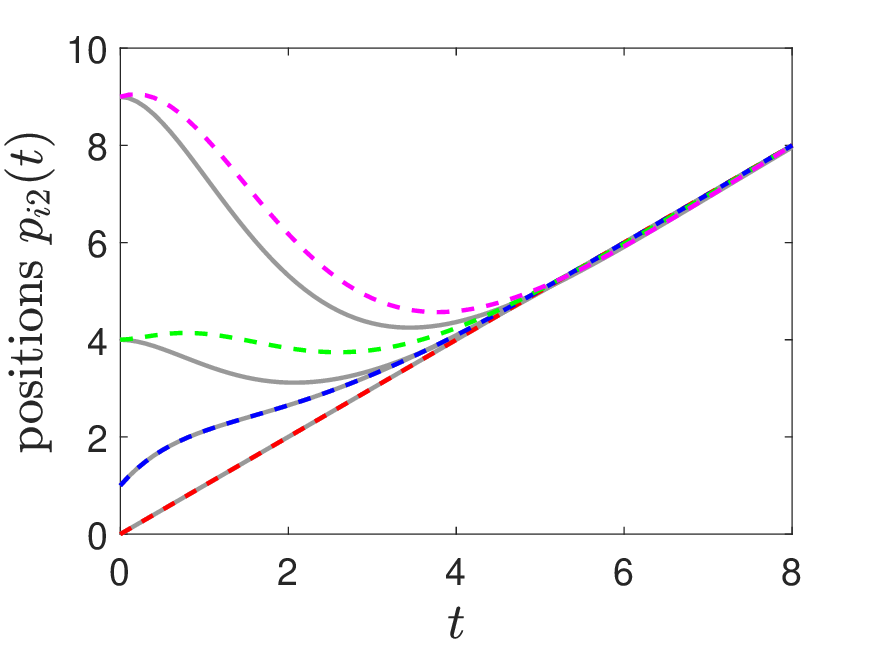}
   \includegraphics[width=0.25\textwidth]{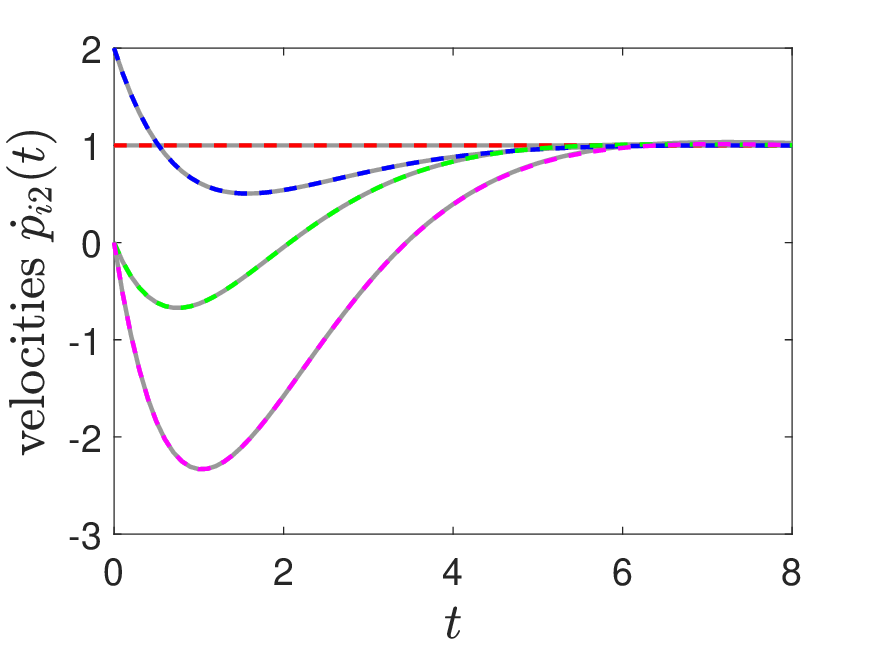}
   \includegraphics[width=0.25\textwidth]{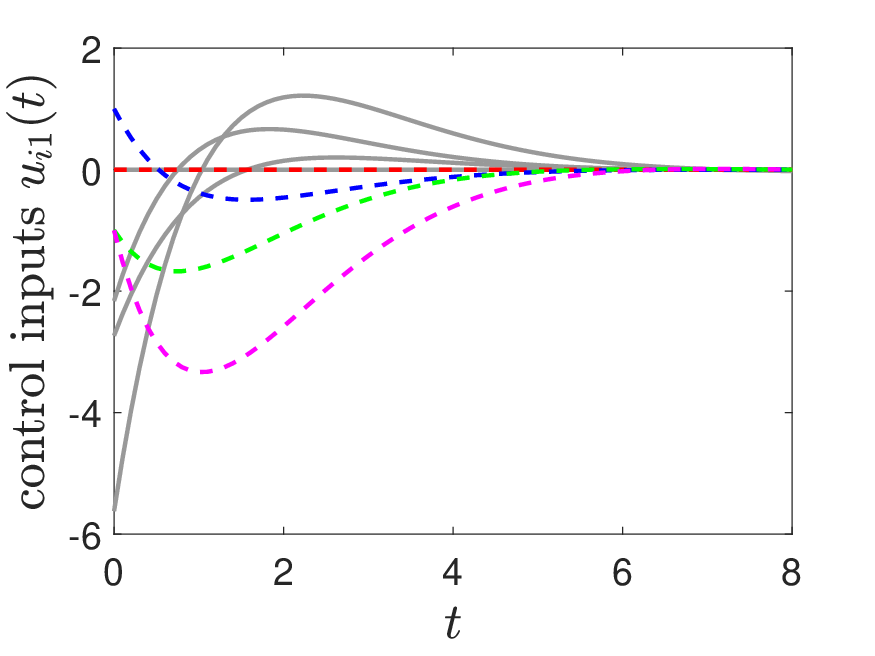}
\caption{Delay-free time histories under the non-distributed and distributed solutions. }
\label{fig:hist}
\end{figure*}

\begin{figure*}[ht]
   \centering
   \includegraphics[width=0.25\textwidth]{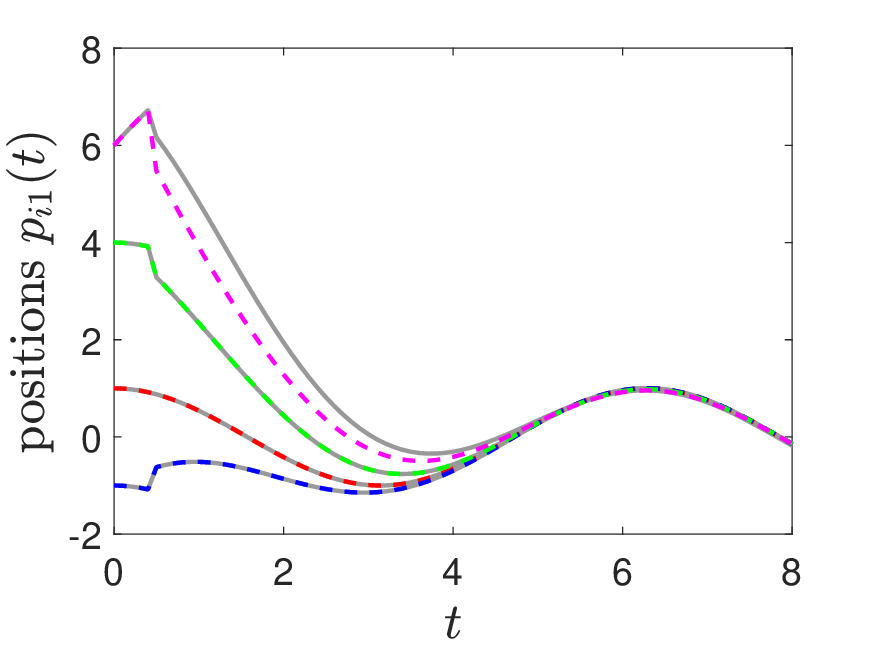}
   \includegraphics[width=0.25\textwidth]{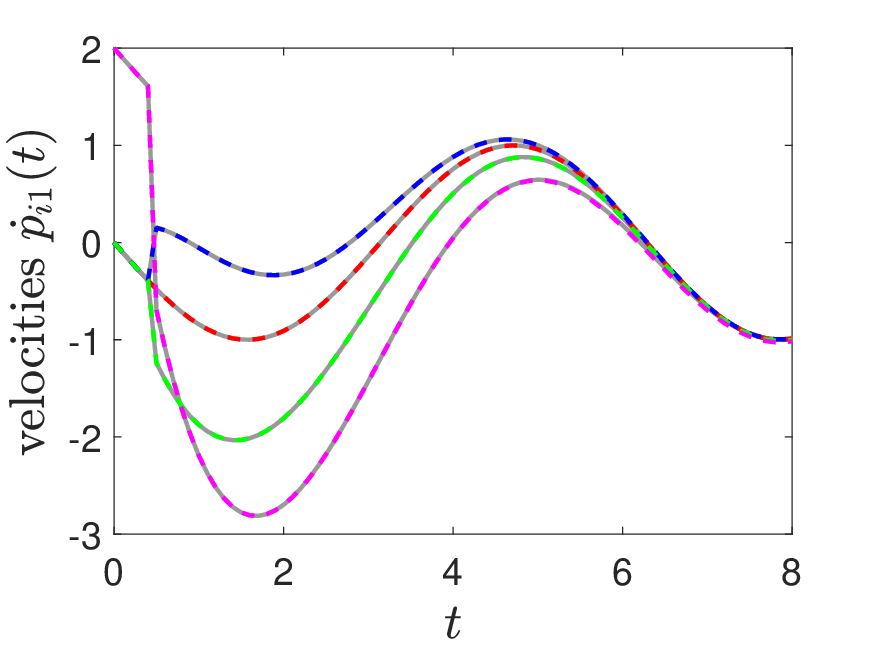}
   \includegraphics[width=0.25\textwidth]{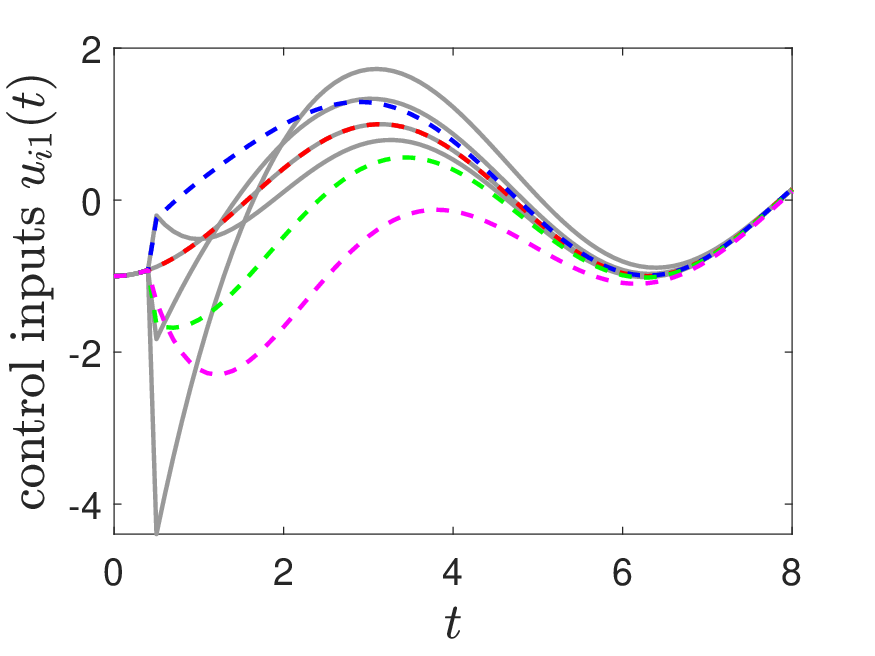}
   \includegraphics[width=0.25\textwidth]{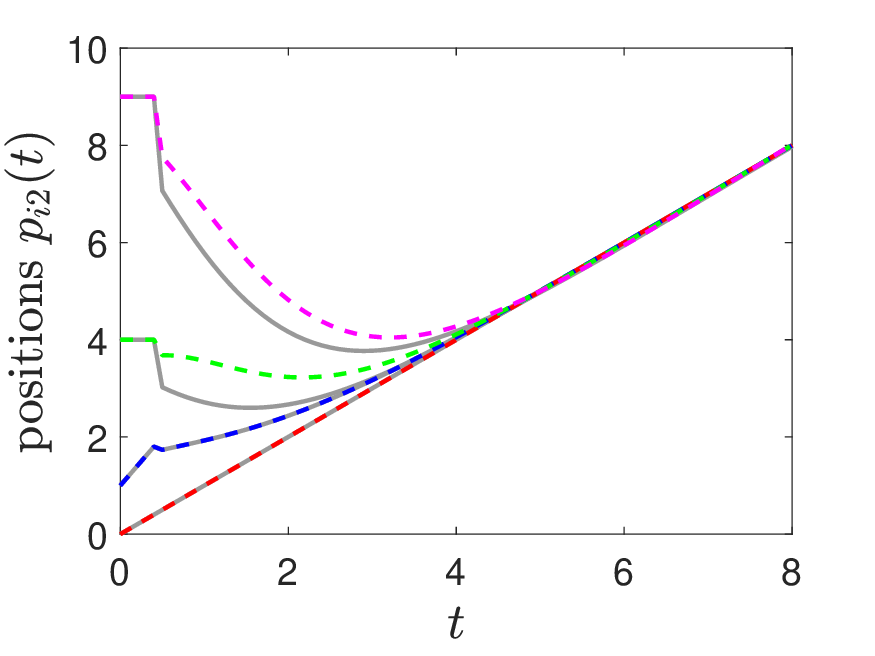}
   \includegraphics[width=0.25\textwidth]{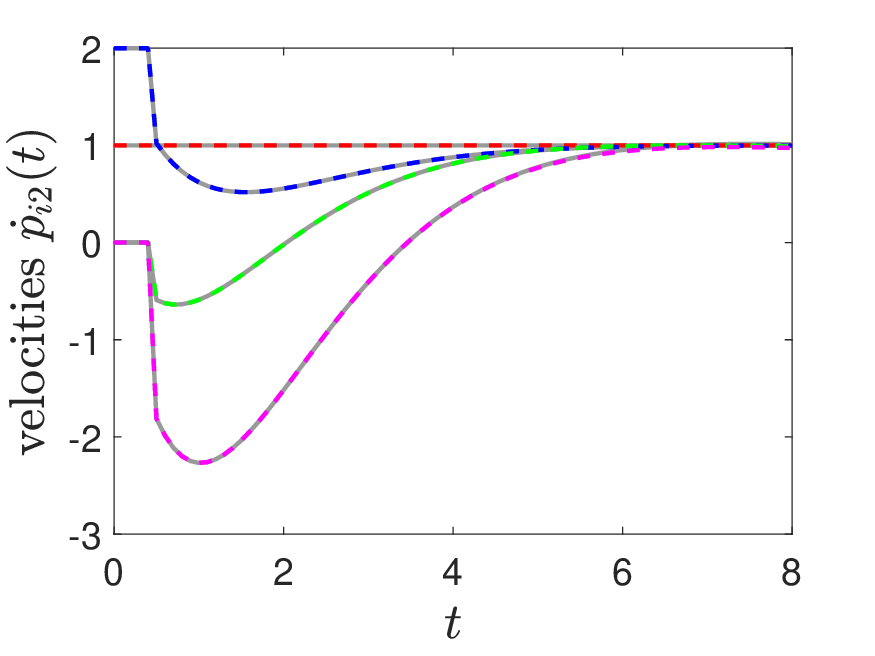}
   \includegraphics[width=0.25\textwidth]{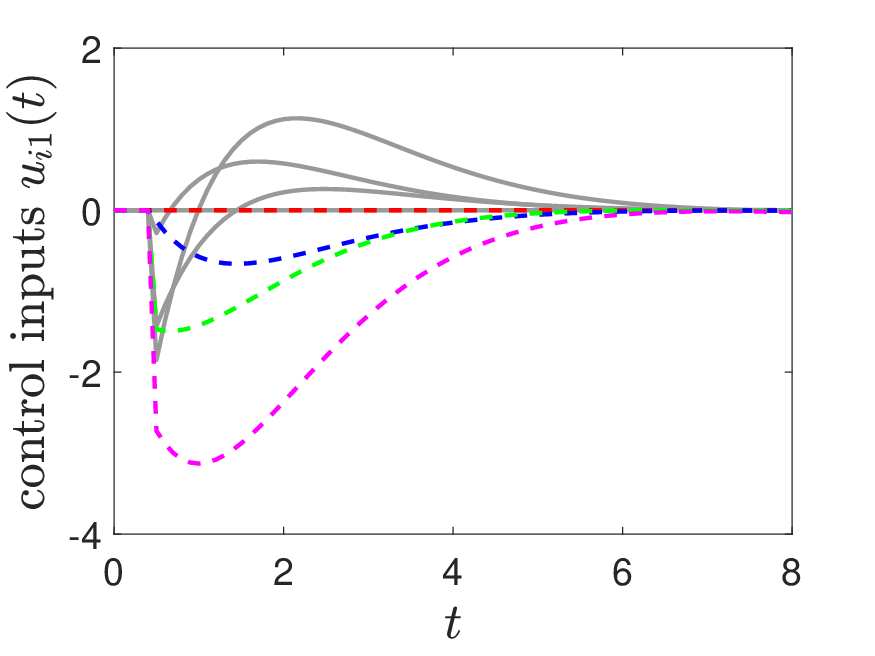}
\caption{Delayed ($\tau=0.5$) time histories under the non-distributed and distributed solutions.} 
\label{fig:histdelay}
\end{figure*}

\section{Conclusion}\label{sec:con}
This paper presents a fully distributed explicit solution for the finite-planning horizon DGGs defined to govern the cooperative control of double-integrator multi-agent systems in a $q$-dimensional space ($q\in\{1,2,3\}$). The classical solution is associated with solving either a set of coupled (asymmetric) Riccati differential equations or, equivalently, a TPBVP and is non-distributed since it requires global state information. For future work, the proposed systematic approach could be generalized to solving DGGs with the feedback information structure and/or MAS with more complex dynamics.

\appendices

Appendixes, if needed, appear before the acknowledgment.

\counterwithin{equation}{section}
\section{Decomposition of \textit{M}}\label{app:decompose}

The matrix $M$ is defective (see Appendix \ref{app:defective}). Its Jordan decomposition is given by (\ref{eq:jordan}) where 
\begin{equation}
    J=\begin{bmatrix}
        I_{m(m-1)\times m(m-1)} & & &\\
        & J_2 (0) &  & \\
        & & \ddots & \\
        & &  & J_1 (\lambda_i )
    \end{bmatrix} 
\end{equation}
is the Jordan normal form of $M$ with $J_i(.)$ defined in (\ref{eq:jordan-block}) being the Jordan block of size $i$ (see Appendix \ref{app:jordan}). 
The generalized modal matrices $\Phi$ and $\Psi$ are constituted from vectors $v_i$ and $w_i$ as following 
\begin{equation*}
\Phi=\begin{bmatrix}
  0_{2m\times2m(m-1)}  & \cdots, v_i, \cdots \\
  *
\end{bmatrix}, 
\Psi=\begin{bmatrix}
  0_{2m(m-1)\times 2m},* \\
  \vdots\\ w_i\\ \vdots 
\end{bmatrix}  
\end{equation*}
where * are the elements/blocks not to be concerned with (see Appendix \ref{app:eigenvector}). 

Vectors $v_i$ and $w_i$, given at the top of the next page,
\begin{figure*}%[!b]%[t!]
\small
\begin{align}
  v_i(\lambda_i) & =\left[\left(\sigma_i^\top(\lambda_i),-\lambda_i\sigma_i^\top(\lambda_i) \right) ,(0_{1\times m},0_{1\times m}),\cdots,\left((\frac{1}{\lambda_i}-\tau)\sigma_i^\top(\lambda_i)W_i,\frac{1}{\lambda_i^2}-(\tau^2+1)\sigma_i^\top(\lambda_i)W_i\right) ,\cdots,(0_{1\times m},0_{1\times m}) \right]^\top \label{eq:right-eign}\\
  w_i(\lambda_i) & =\left[(\left(\frac{\mu_i}{\hat{r}_i}\tau\frac{1}{\lambda_i^2}+\frac{\mu_i}{\hat{r}_i}(\tau^2+1)\frac{1}{\lambda_i}-\lambda_i\right)\sigma_i^\top(\lambda_i),\sigma_i^\top(\lambda_i) ),(0_{m\times 1},0_{m\times 1}),\cdots,(\frac{1}{\lambda_i^2}\sigma_i^\top(\lambda_i) \delta_i,\frac{1}{\lambda_i}\sigma_i^\top(\lambda_i) \delta_i),\cdots,(0_{m\times 1},0_{m\times 1}) \right] \label{eq:left-eign}
\end{align}
\normalsize
\hrule
\end{figure*}
are the right and left eigenvector associated with the eigenvalue $\lambda_i$ of $M$, respectively (see Appendix \ref{app:eigenvector}). 

The nonzero eigenvalues of $M$ are $\lambda_i,\overline{\lambda}_i,-\lambda_i,-\overline{\lambda}_i,\cdots$ where
  \begin{equation}\label{eq:quartic-roots}
  \lambda_i=\sqrt{\frac{1}{2}\frac{\mu_i}{\hat{r}_i}(\tau^2+1)+\sqrt{\frac{1}{4}\left(\frac{\mu_i}{\hat{r}_i}(\tau^2+1)\right)^2-\frac{\mu_i }{\hat{r}_i}}} 
  \end{equation}
  for $i=1,\cdots,m$ and $\overline{\lambda}_i$ denotes the conjugate of $\lambda_i$ (see Appendix \ref{app:eigenvalue}).

\section{Eigenvalues of \textit{M}}\label{app:eigenvalue}
Eigenvalues of $M$ correspond to the roots of its characteristic polynomial
\begin{equation}
    \rho(\lambda)=\det(M-\lambda I)
\end{equation}
where $\lambda$ denotes the unknown eigenvalues. $M=\begin{bmatrix}
    -\hat{A} & S \\ Q & V
\end{bmatrix}$ where $V=\mathrm{diag}(\hat{A}^\top,\cdots,\hat{A}^\top)$, $S=[S_1,\cdots,S_m]$, and $Q=[Q_1,\cdots,Q_m]^\top$.

The determinant of $M-\lambda I$ is~\cite{Silvester2000DeterminantsOB}
\begin{align}\label{eq:block-det}
     &\det(M-\lambda I)=\det \begin{bmatrix}
     -\hat{A}-\lambda I & S\\
     Q & V-\lambda I
   \end{bmatrix}\nonumber\\
   &=\det(-\hat{A}-\lambda I)  \det(V-\lambda I-Q(-\hat{A}-\lambda I)^{-1}S)
\end{align}
where using the Kronecker property (\ref{eq:Kronecker})
\begin{align}\label{eq:block-det2}
     &\det(\hat{A}-\lambda I_{2m\times 2m})=\det(- 
     \begin{bmatrix}
     0 & 1\\ 0 & 0
   \end{bmatrix} \otimes I_{m\times m} - \lambda I_2 \otimes I_{m\times m}) \nonumber\\
   &=\det(- 
     \begin{bmatrix}
     \lambda & 1\\
     0 & \lambda
   \end{bmatrix} \otimes I_{m\times m})= \Big(\det (-\begin{bmatrix}
     \lambda & 1\\
     0 & \lambda
   \end{bmatrix})\Big)^{m}= \lambda^{2m} 
\end{align}
Also,
\begin{align} \label{eq:block-inverse}
     (-&\hat{A}-\lambda I_{2m\times 2m})^{-1}=(- 
     \begin{bmatrix}
     \lambda & 1\\
     0 & \lambda
   \end{bmatrix} \otimes I_{m\times m} )^{-1}\nonumber\\
   &=- 
     \begin{bmatrix}
     \lambda & 1\\
     0 & \lambda
   \end{bmatrix}^{-1} \otimes I_{m\times m} =\begin{bmatrix}
     -\frac{1}{\lambda} & \frac{1}{\lambda^2}\\
     0 & -\frac{1}{\lambda}
   \end{bmatrix} \otimes I_{m\times m}.
\end{align}
From (\ref{eq:Qi})  
\begin{align}\label{eq:Q-exp}
    Q_i&=\begin{bmatrix}
       I & 0 \\ \tau I & I 
    \end{bmatrix}\begin{bmatrix}
       W_i & 0 \\ 0 & W_i
    \end{bmatrix}\begin{bmatrix}
       I & \tau I \\ 0 & I 
    \end{bmatrix}=\begin{bmatrix}
       1 & \tau \\ \tau & \tau^2+1 
    \end{bmatrix}\otimes W_{i}
\end{align}
and from (\ref{eq:S}) $\delta_i=\frac{1}{\hat{r}_i}\hat{b}_i\hat{b}_i^\top$.

Using the expressions above, we have  
\begin{align}
   Q_i (-\hat{A}-\lambda I)^{-1}S_i&=
   \left(\begin{bmatrix}
       1 & \tau \\ \tau & \tau^2+1 
    \end{bmatrix}\otimes W_{i} \right)
    \nonumber\\
    &\left(\begin{bmatrix}
     -\frac{1}{\lambda} & \frac{1}{\lambda^2}\\
     0 & -\frac{1}{\lambda}
   \end{bmatrix} \otimes I_{m\times m} \right)
       \left(\begin{bmatrix}
     0 & 0\\
     0 & 1
   \end{bmatrix} \otimes \delta_i \right)\nonumber\\
    &
   = \begin{bmatrix}
       0 & \frac{1}{\lambda^2}-\tau\frac{1}{\lambda} \\ 0 & \tau\frac{1}{\lambda^2}-(\tau^2+1)\frac{1}{\lambda} 
    \end{bmatrix}\otimes W_{i}\delta_i
\end{align}
where
$W_i\delta_i=\mathrm{diag}(0,\cdots,\frac{\mu_i}{\hat{r}_i},\cdots,0) $ and $W_i\delta_j=0~\forall i\neq j$. Similarly,
\begin{align}\label{eq:block-inverse2}
  Q(-\hat{A}&-\lambda I)^{-1}S=\nonumber\\
    &\mathrm{diag}(\cdots,\begin{bmatrix}
       0 & \frac{1}{\lambda^2}-\tau\frac{1}{\lambda} \\ 0 & \tau\frac{1}{\lambda^2}-(\tau^2+1)\frac{1}{\lambda} 
    \end{bmatrix}\otimes W_i\delta_i,\cdots).
\end{align}

%As $\hat{\bm{A}}^\top- \lambda \bm{I}_{6m\times 6m}=\bm{a}^\top\otimes\bm{I}_{m\times m} - \lambda\bm{I}_{2\times 2} \otimes\bm{I}_{3\times 3}\otimes\bm{I}_{m\times m} = (\begin{bmatrix}
%    0 & 0 \\1 & 0
%\end{bmatrix}\otimes\bm{I}_{3 \times 3}-\lambda\bm{I}_{2\times 2}\otimes\bm{I}_{3\times 3})\otimes\bm{I}_{m\times m} = \begin{bmatrix}
%  -\lambda & 0\\
%  1 & -\lambda
%  \end{bmatrix}\otimes\bm{I}_{3\times 3}\otimes\bm{I}_{m\times m}= \begin{bmatrix}
%  -\lambda & 0\\
%  1 & -\lambda
%  \end{bmatrix}\otimes\bm{I}_{m\times m}\otimes\bm{I}_{3\times 3}$ 

On the other hand,
\begin{align}\label{eq:block-op}
V- \lambda I&=\mathrm{diag}(\cdots,\hat{A}^\top- \lambda I_{2m\times 2m},\cdots)\nonumber\\
    &=\mathrm{diag}(\cdots,\begin{bmatrix}
  -\lambda & 0\\
  1 & -\lambda
  \end{bmatrix}\otimes I_{m\times m},\cdots).
\end{align}
Substituting (\ref{eq:block-det2}), (\ref{eq:block-inverse2}) and (\ref{eq:block-op}) in (\ref{eq:block-det}) yields (\ref{eq:quartic}), shown at the top of the next page. 
\begin{figure*}%[!b]%[t!]
\begin{align}\label{eq:quartic}
     \rho(\lambda)=&
     \lambda^{2m} \det\big(\mathrm{diag}(\cdots,\begin{bmatrix}
  -\lambda I_{m\times m} & (-\frac{1}{\lambda^2}+\tau\frac{1}{\lambda}) W_i\delta_i\nonumber\\
  I_{m\times m} & -\lambda I_{m\times m}+((\tau^2+1)\frac{1}{\lambda}-\tau\frac{1}{\lambda^2})W_i\delta_i
  \end{bmatrix}, \cdots)\big)\nonumber\\
  =&\lambda^{2m} \prod_{i=1}^{m}\det\Bigg(\begin{bmatrix}
  -\lambda I_{m\times m} & (-\frac{1}{\lambda^2}+\tau\frac{1}{\lambda}) W_i\delta_i\nonumber\\
  I_{m\times m} & -\lambda I_{m\times m}+((\tau^2+1)\frac{1}{\lambda}-\tau\frac{1}{\lambda^2})W_i\delta_i
  \end{bmatrix}\Bigg)\nonumber\\
  = &\lambda^{2m} \prod_{i=1}^{m}\det (-\lambda I_{m\times m})\prod_{i=1}^{m}\det\left(\left((\tau^2+1)\frac{1}{\lambda}-\frac{1}{\lambda^3}\right)W_i\delta_i-\lambda I_{m\times m}\right)\nonumber \\
  = &\lambda^{m(m+2)} \prod_{i=1}^{m}\det \left(-\lambda I_{(m-1)\times (m-1)}\right)\prod_{i=1}^{m}\left(\left((\tau^2+1)\frac{1}{\lambda}-\frac{1}{\lambda^3}\right)\frac{\mu_i}{\hat{r}_i}-\lambda\right)\nonumber \\
  = &\lambda^{m(2m+1)} \prod_{i=1}^{m}\frac{1}{\lambda^3}\left(\frac{\mu_i}{\hat{r}_i}\Big((\tau^2+1)\lambda^2-1\Big)-\lambda^4\right) = \lambda^{2m(m-1)} \prod_{i=1}^{m}\left(\lambda^4-\frac{\mu_i}{\hat{r}_i}(\tau^2+1)\lambda^2+\frac{\mu_i}{\hat{r}_i}\right)
\end{align}
\hrule
\end{figure*}
The quartic polynomial in (\ref{eq:quartic}) has four roots given by (\ref{eq:quartic-roots}), two opposites in sign, and two conjugates.

\section{Eigenvectors of \textit{M}}\label{app:eigenvector}
\subsection*{Eigenvectors associated with the nonzero eigenvalues:}\label{app:nonzero-associated-eigvec}
%Since for the nonzero eigenvalue $\zeta=1$, we must find an eigenvector of rank $1$. 
Matrix $M$ consists of $2(m+1)\times 2(m+1)$ blocks of size $m\times m$ as follows
\begin{equation}\label{eq:Mblock}
    M=\begin{bmatrix}
    0_{m\times m} & -I & 0 & 0  & \cdots & 0 & 0 \\
    0 & 0_{m\times m} & 0 & \delta_1 & \cdots & 0 & \delta_m\\
    W_1 & \tau W_1  & 0 & 0 & \cdots & 0 & 0 \\
    \tau W_1 & (\tau^2+1) W_1 & I & 0 & \cdots & 0 & 0 \\
    \vdots & \vdots &  &  & \ddots & & \\
    W_m & \tau W_m  & 0 & 0 & \cdots & 0 & 0 \\
    \tau W_m & (\tau^2+1) W_m & 0 & 0 & \cdots & I & 0 \\
\end{bmatrix}.
\end{equation}
Any nonzero right eigenvector $v_i$ satisfies  
\begin{equation}\label{eq:right-eigenvector}
    (M-\lambda_i I)v_i=0.
\end{equation}
Substituting (\ref{eq:right-eign}) into the left-hand side of (\ref{eq:right-eigenvector}) yields
\begin{align}
       (M-\lambda_i I)v_i&=\Bigg[\left(0_{1\times m},\lambda_i^2 + \frac{\mu_i}{\hat{r}_i}\left(\frac{1}{\lambda_i^2}-(\tau^2+1)\right) \sigma_i^\top(\lambda_i)\right),\nonumber\\ & (0_{1\times m}, 0_{1\times m}), \cdots,  (0_{1\times m}, 0_{1\times m}) \Bigg]^\top.
\end{align}
Here, the expression $\lambda_i^2 + \frac{\mu_i}{\hat{r}_i}\left(\frac{1}{\lambda_i^2}-(\tau^2+1)\right)$ appears to be the quartic polynomial in (\ref{eq:quartic}) that $ \lambda_i$ is a zero of it. Therefore, (\ref{eq:right-eign}) satisfies (\ref{eq:right-eigenvector}). 

Any nonzero left eigenvector $w_i$ satisfies
\begin{equation}\label{eq:left-eigenvector}
   w_i (M-\lambda_i I)=0.
\end{equation}
Similarly, substituting (\ref{eq:left-eign}) into (\ref{eq:left-eigenvector}), (\ref{eq:left-eigenvector}) is satisfied. By adopting $\varpi_i(\lambda_i)$ as (\ref{eq:varpi}), $v_i(\lambda_i)$ and $w_i(\lambda_i)$ are normalized so that $w_i(\lambda_i) v_i(\lambda_i)=1$.

\subsection*{Eigenvectors associated with the zero eigenvalue:}\label{app:zero-associated-eigvec}
  
From Appendix \ref{app:eigenvalue} we notice that $M$ has $2m(m-1)$ zero eigenvalues. Thereby, the definition of the right and left eigenvectors in (\ref{eq:right-eigenvector}) and (\ref{eq:left-eigenvector}), respectively, reduces to 
\begin{equation}\label{eq:zero-eigenvector}
    Mv_0=0, \quad w_0M=0.
\end{equation}

Vectors
\begin{align*}
    v_0&=[(0_{1\times m},0_{1\times m}),(0_{1\times m},\vartheta_1^\top ),\cdots,(0_{1\times m},\vartheta_m^\top ) ]^\top\\
    w_0&=[(0_{1\times m},0_{1\times m}),(\vartheta_1^\top,0_{1\times m}),\cdots,(\vartheta_m^\top,0_{1\times m}) ]
\end{align*}
where $\vartheta_i=[\vartheta_1^i, \cdots, [0]_i,\cdots,\vartheta_m^i]^\top \in \mathbb{R}^m$ satisfy (\ref{eq:zero-eigenvector}).

From Appendix \ref{app:jordan}, we notice that the Jordan blocks associated with the zero eigenvalues are of size $2$ where each corresponds to a chain of generalized eigenvectors of rank $2$. A generalized eigenvector of rank $1$ is an ordinary eigenvector.  

The generalized right and left eigenvectors of rank $2$ satisfy
\begin{equation}\label{eq:right-eign-vec}
    M \hat{v}_0=v_0, \quad \hat{w}_0 M =w_0
\end{equation}
where $\hat{v}_0\neq 0$ and $\hat{w}_0\neq 0$. It is verified that $\hat{v}_0=w_0^\top$ and $\hat{w}_0=v_0^\top$. The right and the right generalized eigenvector, as well as the left and the left generalized eigenvector associated with the zero eigenvalues, have the form $[0_{1\times 2m},*]^\top$ and $[0_{1\times 2m},*]$, respectively.

\section{\textit{M} Is Defective}\label{app:defective}
For a defective matrix, its Jordan normal form is a block-diagonal matrix whose diagonally located blocks are associated with the eigenvalue of that matrix in the form of 
\begin{equation}\label{eq:jordan-block}
    J_k(\lambda_i)=\begin{bmatrix}
        \lambda_i & 1 & 0 & & \\
        0 & \lambda_i & 1 & 0 \\
        &  &  \ddots &\ddots & \\
        &0 &  & \lambda_i & 1 \\
        &  &  &  0  &  \lambda_i                      
    \end{bmatrix}_{k \times k}.
\end{equation}

To determine whether $M$ is defective or not, we inspect the geometric multiplicity associated with the zero eigenvalues. The number of linearly independent eigenvectors associated with an eigenvalue is its geometric multiplicity. A defective matrix has an eigenvalue with its geometric multiplicity less than its algebraic multiplicity~\cite{shores2007applied}.

Let $\mathcal{N}(M-\lambda I)$ be the dimension of the null space of $M-\lambda I$. The geometric multiplicity of $\lambda$ is $\mathrm{dim}\mathcal{N}(M-\lambda I)$ which is the number of free variables in
\begin{equation}\label{eq:eigenvector}
    (M-\lambda I)v=0
\end{equation}
where $v$ is a nonzero vector.

Consider (\ref{eq:Mblock}). Let $v=[(\varrho_0^\top,\vartheta_0^\top),(\varrho_1^\top,\vartheta_1^\top),\cdots,(\varrho_m^\top,\vartheta_m^\top)]^\top$, $\varrho_i=[\varrho_1^i,\cdots,\varrho_m^i]^\top\in \mathbb{R}^{m\times 1}$, $\vartheta_i=[\vartheta_1^i,\cdots,\vartheta_m^i]^\top\in \mathbb{R}^{m\times 1}$. For $\lambda=0$, (\ref{eq:eigenvector}) reduces to
\begin{equation}
    \vartheta_0=\varrho_i=0,\quad i=0,\cdots,m,\quad \sum_{i=1}^{m}\delta_i\vartheta_i= \begin{bmatrix}
    r_1^{-1}\varrho_1^1\\
    \vdots\\
    r_m^{-1}\varrho_m^m
    \end{bmatrix}=0.
\end{equation}
Clearly, $\vartheta_i^i$ in $\vartheta_i$ must be zero, thus, there are $m-1$ free variables left in $\vartheta_i$. Consequently, the total number of free variables is $m(m-1)$ and 
\begin{equation*}
    \mathrm{dim}\mathcal{N}(M)=m(m-1).
\end{equation*}
As a result of the geometric multiplicity of the zero eigenvalues being less than their algebraic multiplicity, $M$ is defective.

\section{Jordan Blocks of \textit{M}}\label{app:jordan}

By definition, the size of the Jordan block associated with $\lambda_i$, i.e., $\zeta$, is the first integer for which $\mathrm{dim}\mathcal{N}(M-\lambda_i I)^\zeta$ stabilizes. This is equal to the number of free variables in 
\begin{equation}\label{eq:free-variable}
    (M-\lambda I)^j v=0.
\end{equation}

For the eigenvectors associated with the zero eigenvalues, we have $\mathrm{dim}\mathcal{N}(M)^1=m(m-1)$. It can be verified that vector $v$ (as defined in Appendix \ref{app:defective}) with $\varrho_0=\vartheta_0=0$ and $\varrho_i=\vartheta_i=[\varrho_1^i,\cdots,[0]_i,\cdots,\varrho_m^i ]^\top$, $ i=1,\cdots,m$ satisfies (\ref{eq:free-variable}) for $j=2$ and $j=3$. Therefore, $\mathrm{dim}\mathcal{N}(M)^2=\mathrm{dim}\mathcal{N}(M)^3=2m(m-1)$. It is seen that $\zeta=2$ and thus, there are $m(m-1)$ Jordan blocks of size $2$.

For the eigenvectors associated with nonzero eigenvalues, the number of free variables in the equation (\ref{eq:free-variable}) for $j=1$ and $j=2$ are both equal to $m$. Therefore, $\zeta=1$ and the corresponding Jordan block is of size $1$.

\section{Decomposition of \textit{H}}\label{app:decomposeH}
Substituting $\mathrm{e}^{\phi M}$ from (\ref{eq:matrix-exponential}) shows that
\begin{equation}
 H(\phi)=\begin{bmatrix}I_{2m\times 2m}& 0& \cdots &0\end{bmatrix} \Phi \mathrm{e}^{\phi J} \Psi 
\begin{bmatrix}
I\\
Q_{1T}\\
\vdots\\
Q_{mT} 
\end{bmatrix}.
\end{equation}
This expression is broken apart into the expressions (\ref{eq:part1})-(\ref{eq:part3}).
\begin{align}\label{eq:part1}
&\begin{bmatrix}I_{2m\times 2m}& 0& \cdots &0\end{bmatrix} \Phi=\nonumber\\
&\qquad\qquad\begin{bmatrix}0_{2m\times 2m(m-1)} \cdots,\begin{pmatrix} 1 \\-\lambda_i \end{pmatrix}\otimes \sigma_i(\lambda_i),\cdots \end{bmatrix} 
\end{align}
\begin{equation}\label{eq:part2}
   \mathrm{e}^{\phi J}= \begin{bmatrix}
    I_{m(m-1)\times m(m-1)} \otimes \mathrm{e}^{\phi J_2 (0)} & 0 & \cdots &0     \\
   0&  \ddots &0 & \vdots \\
   \vdots& \cdots & \mathrm{e}^{\phi \lambda_i}& 0 \\
   0 & \cdots &0 & \ddots
\end{bmatrix}
\end{equation}
\begin{figure*}%[!b]%[t!]
\begin{equation}\label{eq:part3}
\Psi 
\begin{bmatrix}
I\\
Q_{1T}\\
\vdots\\
Q_{mT} 
\end{bmatrix}= \begin{bmatrix} 
        * \\
        \vdots\\
        \left((\frac{\omega_i}{\hat{r}_i}+\frac{\mu_i}{\hat{r}_i}\tau)\frac{1}{\lambda_i^2}+(\frac{\omega_i}{\hat{r}_i}\tau+\frac{\mu_i}{\hat{r}_i}(\tau^2+1))\frac{1}{\lambda_i}-\lambda_i\right)\sigma_i^\top(\lambda_i)\qquad 
        \left(\frac{\omega_i}{\hat{r}_i}\tau\frac{1}{\lambda_i^2}+\frac{\omega_i}{\hat{r}_i}(\tau^2+1)\frac{1}{\lambda_i}+1\right)\sigma_i^\top(\lambda_i) \\
        \vdots
     \end{bmatrix} 
\end{equation}
\hrule
\end{figure*}
Finally, by multiplying the right-hand sides of (\ref{eq:part1})-(\ref{eq:part3}) we get (\ref{eq:expanH}).

%\begin{IEEEbiography} [{\includegraphics[width=1in,height=1.25in,clip,keepaspectratio]{jond}}]{Hossein B. Jond} received a Ph.D. degree in computer engineering from KTU, Trabzon, Turkey, in 2019. He was a post-doctoral researcher from 2019 to 2021 with the Department of Computer Science at the VSB-TUO, Ostrava, Czechia, where he is currently an assistant professor. His research interests include game theory and dynamic optimization. Dr. Jond was a recipient of the Turkish Government Success Scholarship in 2015 and the Best Student Paper Award at the 21st International Student Conference on Electrical Engineering held in Prague in 2017. \end{IEEEbiography}

\end{document}